\documentclass[sn-mathphys-ay]{sn-jnl}

\usepackage{graphicx}%
\usepackage{multirow}%
\usepackage{amsmath,amssymb,amsfonts}%
\usepackage{amsthm}%
\usepackage{mathrsfs}%
\usepackage[title]{appendix}%
\usepackage{xcolor}%
\usepackage{textcomp}%
\usepackage{manyfoot}%
\usepackage{booktabs}%
\usepackage{algorithmicx}%
\usepackage{algpseudocode}%
\usepackage{listings}%
\usepackage{algorithm2e}

\usepackage{mathtools}

\newtheorem{assumption}{Assumption}
\newtheorem{lemma}{Lemma}

\DeclareMathOperator*{\argmin}{argmin}
\def \b1{{\bf 1}}

\def \ba{{\bf a}}

\def \bX{{\bf X}}

\def \bD{{\bf D}}

\def \mO{{\mathcal O}}

\def \bbeta{{\boldsymbol \beta}}

\def \P{{\mathbb{P}}}

\def \mE{{\mathbb{E}}}
\def \mI{{\mathbb{I}}}
\def \mR{{\mathbb{R}}}
\def \I{{\mathbb{I}}}
\def \tr{{\rm tr}}
\def \d{{\rm d}}

\def \Var{\mathbb{V}{\rm ar}}
\def \Cor{\mathbb{C}{\rm or}}

\def \F{{\mathcal F}}
\def \H{{\mathcal H}}

\theoremstyle{thmstyleone}%
\newtheorem{theorem}{Theorem}
\newtheorem{proposition}{Proposition}%

\theoremstyle{thmstyletwo}%

\theoremstyle{thmstylethree}%

\begin{document}

\title[Online robust estimation and bootstrap inference for function-on-scalar
regression]{Online robust estimation and bootstrap inference for function-on-scalar
regression}


\author[1]{\fnm{Guanghui} \sur{Cheng}}\email{chenggh845@nenu.edu.cn}

\author*[2]{\fnm{Wenjuan} \sur{Hu}}\email{huwj183@cueb.edu.cn}

\author*[3]{\fnm{Ruitao} \sur{Lin}}\email{rlin@mdanderson.org}

\author[4]{\fnm{Chen} \sur{Wang}}\email{rstacw@hku.hk}

\affil[1]{ \orgname{Guangzhou Institute of International Finance, Guangzhou University}, \city{Guangzhou}, \postcode{500016}, \country{China}}

\affil[2]{\orgname{Department of Statistics, Capital University of Economics and Business}, \city{Beijing}, \postcode{100071}, \country{China}}

\affil[3]{\orgname{Department of Biostatistics, The University of Texas MD Anderson Cancer Center}, \city{Houston}, \postcode{77030}, \country{USA}}

\affil[4]{\orgname{Department of Statistics and Actuarial Science, The University of Hong Kong}, \city{Hong Kong}, \postcode{999077}, \country{Hong Kong}}


\abstract{We propose a novel and robust online function-on-scalar regression technique via geometric median to learn associations between functional responses and scalar covariates based on massive or streaming datasets. The online estimation procedure, developed using the average stochastic gradient descent algorithm, offers an efficient and cost-effective method for analyzing sequentially augmented datasets, eliminating the need to store large volumes of data in memory. We establish the almost sure consistency, $L_p$ convergence, and asymptotic normality of the online estimator. To enable efficient and fast inference of the parameters of interest, including the derivation of confidence intervals, we also develop an innovative two-step online bootstrap procedure to approximate the limiting error distribution of the robust online estimator. Numerical studies under a variety of scenarios demonstrate the effectiveness and efficiency of the proposed online learning method. A real application analyzing  PM$_{2.5}$ air-quality data is also included to exemplify the proposed online approach.}




\maketitle

\section{Introduction}\label{Introduction}
\subsection{Function-on-scalar regression}
Functional data analysis (FDA) is an effective way for modeling  high-dimensional or  potentially infinite-dimensional datasets and can find applications in may real-life situations. 
One of the most widely studied approaches in the field of FDA, as seen in the literature, is the scalar-on-function regression, where the response is a scalar and the covariates are functional data \citep{Ramsay2005,Cai2006, Morris2015, Wang2017}. In the analysis of gene expression and imaging data, the function-on-scalar regression, which involves regressing a functional response on scalar predictors, has been gaining popularity \citep{Reiss2010,Goldsmith2015,Goldsmith2016,Fan2017,Bauer2018,Ghosal2023}.   
 Formally, the function-on-scalar regression model (FSRM) can be expressed as  
\begin{align}\label{functional_model}
Y(t)=\bX^{T}\bbeta(t)+ U(t)
\end{align}
where $t$ is the functional index, $Y(t)$ is a functional response on a compact support ${\mathcal T}$, and $\bX=(X_1, \ldots, X_d)^{T}$ are scalar covariates in $\mR^d$,  $\bbeta(t)=(\beta_1(t), \ldots, \beta_d(t))^T$ is the slope function that of primary interest, 
and $U(t)$ is an residual process on ${\mathcal T}$ that is independent of the predictors $\bX$. 
Without loss of generality, we assume ${\mathcal T}=[0,1]$. If the support ${\mathcal T}=\{1:p\}$,  model (\ref{functional_model}) reduces to the conventional multivariate linear regression with a $p$-dimensional response, and $\bbeta$ is thus a $d \times p$ matrix. Throughout the paper, with a slight abuse of notation, we will omit ``$(t)$'' when referring to functional processes. For instance, we will use $\bbeta$ to denote the entire slope function $\bbeta(t)$ over $t \in \mathcal{T}$.

As a typical motivating application, \citet{Zhang2017} established an air pollution monitoring network in Beijing as part of the national monitoring effort, beginning in January 2013. This network can collect more than 4 million hourly observations from various air-quality monitoring sites across the city.   
In this particular example, the functional responses $Y(t)$ can be a particular measurement of interest quantifying the air quality at different time points $t$ (say hourly), and the scalar covariates can be some daily environmental indicators such as temperature and wind speed. Hence, by performing the function-on-scalar regression model  using model (\ref{functional_model}), one can study the functional association between the hourly air-quality measurement and the daily environmental indicators. 

In the era of big data, massive datasets and streaming datasets, where data is continually updated, are commonly encountered, presenting challenges in terms of data storage, computation, and analysis. Our paper contributes to addressing two major practical challenges involved in analyzing massive functional datasets in real-world applications:

 First, the real-world data distributions are usually irregular and may have heavy tails. Using methods developed under Gaussian assumptions may be suboptimal. In the literature, the majority of existing efforts have predominantly focused on functional mean regression, where the residual term  $U(t)$ is assumed to be a zero-mean stochastic process on $\mathcal{T}$. Under such an assumption, the least squares approach can be applied to estimate the slope function $\bbeta$. However, such a method may not be efficient for heavy-tailed or irregular data.   To address the irregularities in the distributions of high-dimensional or functional data, function-on-scalar quantile regression has been developed \citep{Yang2019,Liu2020,Zhang2022}. The quantile-based methods estimate conditional quantiles and enable statistical inference on the entire conditional distribution of the response. 
 
Second, as with many real-world applications, the amount of data is enormous. For instance, in the motivating example, the dataset contains 4 million observations. Loading such a large dataset into memory all at once might be infeasible, especially when the data is collected from a large number of wearable devices. Traditional approaches typically require halting data collection at a specific time point to analyze a static dataset, which is highly inefficient. Furthermore, analyzing the entire dataset simultaneously using a standard computer is often impossible due to limited storage and computing capacity.
 

\subsection{Our contributions}
We address the aforementioned challenges by (1) employing the {\it geometric median} \citep{Cardot2013}, for the first time in the literature, to estimate the unknown function $\bbeta$ for the FSRM given by Equation (\ref{functional_model}), and (2) developing an online estimation and inference procedure to eliminate the need for storing the entire dataset. 
 Furthermore,  for an unobserved  location $t\in {\mathcal T}$ that is not included in the observed sampling points, we introduce an interpolation approach based on the discrete location observations, and derive the   convergence rate. Under scenarios with dense observations, it can be demonstrated that the sampling frequency of the functional data has minimal impact on the convergence rate.

The geometric median is widely recognized as a robust central location parameter in the analysis of high-dimensional data \citep{Minsker2015,Godichon2016,God2019, Li2022}. \cite{Vardi2000} developed an iterative algorithm to compute the geometric median. For a random element in Hilbert space, \cite{Cardot2013} 
developed a faster online gradient algorithm to estimate the geometric median, both demonstrating its almost sure consistency as well as asymptotic normality. However, they did not provide an online inference framework for the geometric median. As an application, 
 \cite{Roberts2017} employed the geometric median on high-quality, large-scale Earth imaging data to produce Earth observation composites. 
 This approach can effectively reduce spatial noises while preserving spectral relationships. However, to the best of our knowledge, there is currently no existing research on applying the geometric median to the FSRM.

  Instead of assuming $U(t)$ has a zero mean, we assume $U(t)$ on $\mathcal{T}$ follows a stochastic process with a  zero {\it geometric median}. Specifically, let ${\mathcal H}$ be a separable Hilbert space
and ${\mathcal H}^d$ be the product space of ${\mathcal H}$.  A commonly encountered space for ${\mathcal H}$ is $L^2[0,1]$, i.e. the outcome $Y(t)$ represents a function of time.
The conditional geometric median of $Y(t)$ given  scalar covariates $\bX$ for a fixed location $t$ is assumed to be   $m(t)=\bX^{T}\bbeta(t)$, where $m(t)$ is defined as 
\begin{equation}\label{gm-definition}
m(t):=\argmin_{\alpha \in {\mathcal H}} \mE_{Y(t)|\bX}\{\|Y(t)-\alpha\|-\|Y(t)\|\}\,,
\end{equation}
with $\|\cdot\|$ and $\langle \cdot \rangle$ denoting the associated norm and inner product in ${\mathcal H}$, respectively. As shown in the subsequent sections, our definition of $m(t)$ via  geometric median  makes use of the
spatial correlation between different locations, and 
enhances the efficiency and robustness of estimation and inference for $\bbeta$ in the FSRM, especially when dealing with massive or streaming datasets. 
Moreover, if $m^v(t)=\bX^{T}\bbeta(t)$ is the geometric quantile which minimizes $\mE_{Y|\bX}\{\|Y-m^v(t)\|-\langle Y-m^v(t), v \rangle\}$, then $m^v(t)$ generalizes (\ref{functional_model}) to the geometric quantile regression. This approach was studied in \cite{Padilla2022} as a special case of ReLU Networks for multivariate responses.

To deal with massive functional data, we  propose an online procedure for estimating the slope function $\bbeta$ based on the functional average stochastic gradient descent (ASGD) algorithm. We establish the $L_p$ and almost sure convergence for the proposed online estimator.
We also investigate the asymptotic behavior of the online estimator, which essentially resembles a Gaussian process. This asymptotic result is general and includes the asymptotic normality of the estimator at a given location $t$ as a special case. Subsequently, we also propose a novel online bootstrap resampling procedure to approximate the limiting error distribution of our proposed estimator, enabling online inference about the slope function $\bbeta$.

Unlike classical batch learning methods, online learning dynamically updates model estimates using only the newly added    data, enabling real-time fast decision-making \citep{Chen2020, Zhu2021, Lee2022, Li2022, Liu2022}. Consequently, online learning presents a crucial and efficient approach for handling and making inferences from sequentially augmented datasets.  In a related work, \cite{Xie2023}  generalized the perturbed SGD in  \cite{Fang2018}  to functional linear regression in reproducing kernel Hilbert spaces, and also established the asymptotic distribution of the proposed point-wise estimator. In fact, our online bootstrap approach possesses unique features that distinguish it from   the perturbed stochastic gradient descent  (SGD) algorithm  for the finite-dimensional setting
in \cite{Fang2018} as well as the infinite-dimensional setting in \cite{Xie2023}.  The resulting
inferential procedure maintains the computational efficiency of the functional ASGD
method.


The remainder of this paper is structured as follows. Section \ref{sec:02} introduces the online estimation algorithm for the slope function $\bbeta$ and establishes its almost sure and $L_p$ convergence.
Section \ref{sec:03} proposes a novel online bootstrap algorithm for approximating the limiting error distribution, enabling online inference.
Section \ref{sec:04} extends the online estimation algorithm to an interpolation-based estimator for the entire slope function and demonstrates its convergence properties.
Section \ref{sec:05} presents a numerical study to evaluate the empirical performance of our proposed procedure.
Section \ref{sec:06} illustrates the feasibility and usefulness of our method through a real-world application.
All technical details are provided in the Appendix.

{\bf Notation.}
Denote the Euclidean norm of a vector $\mathbf{x}=(x_1,\ldots,x_p)^{\top}$ of length $p$ as $\|\mathbf{x}\|=\left(\sum_{j=1}^{p}|x_j|^{2}\right)^{1/2}$. Let $\lambda_{i}(\cdot)$ and $\tr(\cdot)$ represent the $i$th largest  eigenvalue and the trace of a square matrix, respectively. 
For a $d \times d $ symmetric matrix $\mathbf{A}$, we define the operator norm as 
$\|\mathbf{A}\|=\lambda_{1}(\mathbf{A})$, and 
the Frobenius norm of $\mathbf{A}$ as $\|\mathbf{A}\|_{F}=\left\{\tr(\mathbf{A}^2)\right\}^{1/2}$. For any $u\in \H$,
denote $\|u\|$ as the norm in ${\mathcal H}$ and $\|u\otimes v\|$ is operator norm for linear operator for any $u,v\in\H$.
Finally, $a_n\lesssim b_n$  means that $a_n\leq Cb_n$ for a positive constant $C$; and $a_n\asymp b_n$ indicates $a_n\lesssim b_n$ and $b_n\lesssim a_n$. We use notations in the format of $C$ (with or without subscripts)  to represent constants throughout the paper. 

\section{Online estimation of the FSRM}\label{sec:02}
Suppose that for each sample $i$, $i=1,\ldots,n$, we observe $d$ covariates $\bX_i=(X_{i1}, \ldots, X_{id})^{T} \in \mR^d$ and a set of functional responses   $\{Y_i(t), t\in \mathcal{T}\} \in {\mathcal H}$, where the relationship between $\bX_i$ and $Y_i(t)$ can be characterized by model (\ref{functional_model}).
 Based on the definition (\ref{gm-definition}), the geometric median-based estimator of $\bbeta$ can be derived by minimizing  the following empirical convex loss function 
$$
\hat{\bbeta}  =  \argmin_{\bbeta}\sum_{i=1}^{n}\|Y_{i}-\bX_{i}^{T}\bbeta\|.
$$
By applying the functional SGD approach, which can be written as $ \boldsymbol{\beta}_{n+1} = \boldsymbol{\beta} _{n} - \gamma_n \nabla || Y_{n+1} - {\bf X}_{n+1} \beta_n ||$,   the geometric median-based estimator $\hat{\bbeta} $ can be recursively obtained via 
 \begin{align}\label{RM1}
 \bbeta_{n+1}=\bbeta_n+\gamma_n \frac{\bX_{n+1}(Y_{n+1}- \bX_{n+1}^{T}\bbeta_n )} {\|Y_{n+1}-\bX_{n+1}^{T}\bbeta_n \|}\,,
 \end{align}
where $\gamma_n$ is the descent step size. Therefore,   $\bbeta_{n+1}$ is an estimate based on the  Robbins–Monro algorithm.  We then average all the sequential estimates to obtain our final ASGD estimate, that is, ${\bar \bbeta}_n=\frac{1}{n}\sum_{i=1}^n\bbeta_i$. Hence, we have
\begin{align}\label{sq0}
\bar{\bbeta}_{n+1}={\bar\bbeta}_{n}+\frac{1}{n+1}(\bbeta_{n+1}-{\bar\bbeta}_n)\,.
\end{align}
To derive the consistency and asymptotic  results for $\bar{\bbeta}_n$,  we make the following assumptions throughout the paper,
\begin{assumption}\label{as1}
There exist constant $C>0$ such that 
$1/C \leq \lambda_{d}(\Sigma)\leq 
\lambda_{1}(\Sigma) \leq C$ almost surely, where $\Sigma=\mE(\bX \bX^{T})$. 
\end{assumption}
\begin{assumption}\label{as2}
The error process  $U(t)$ in the FSRM (\ref{functional_model}) has a unique geometric median at zero.
\end{assumption}
\begin{assumption}\label{as3}	
The error process $U(t)$ is not strongly  concentrated around a single point; that is, 
There  is a constant $C >0$ such that 
$
\mE\{\|U(t)-h\|^{-2}\}\leq C\,
$
for all $h \in {\mathcal H}$.
\end{assumption}
Assumption \ref{as1}  is a mild condition on the covariates in linear regression, and has been imposed in \citet{Zhang2022} and \cite{Liu2020}.   Assumption \ref{as2} implies that the error process $U(t)$ is not concentrated on a straight line \citep{Cardot2013}; that is,  for all $h \in {\mathcal H}$, there is a $h^{'} \in {\mathcal H}$ such that $\langle h, h^{'}\rangle=0$ and ${\rm Var}(\langle U(t), h^{'}\rangle) > 0$. Condition \ref{as3} 
implicitly forces no atoms  exist in the distribution of $U(t)$, and it is naturally satisfied in ${\mathbb R}^d$ whenever $d \geq 3$, see \cite{Cardot2017}.  Note that $\mE\{\|U(t)-h\|^{-2}\}\leq C$ implies $\mE\{\|U(t)-h\|^{-1}\}\leq \sqrt{C}$ by  H\"{o}lder's inequality. 

In addition, by taking  $\gamma_n=\gamma n^{-\alpha}$, where $\gamma$ is a positive constant, we impose the following assumption for the step size function $\gamma_n$:
\begin{assumption}\label{as4}	
The sequence $\{\gamma_n\}_{n=1}^\infty$  satisfies the standard conditions required for the Robbins--Monro algorithm; that is, $\sum {\gamma_n}=\infty$ 
and $\sum {\gamma^2_n}<\infty$.
\end{assumption}

\subsection{$L^p$ and almost sure convergence of the online estimator}
We first derive  the  convergence  results of the SGD estimator (\ref{RM1}) in $L^2$ and $L^4$ norms.
\begin{theorem}\label{th1}
Under Assumptions \ref{as1}-\ref{as4},   it holds that for all $\theta\in (\alpha, 2\alpha)$,
\begin{align*}
\mE\{\|\bbeta_n-\bbeta\|^2\}={\mathcal O}\bigg(\frac{1}{n^{\alpha}}\bigg)\,,~~
\mE\{\|\bbeta_n-\bbeta\|^4\}=\mO\bigg(\frac{1}{n^{\theta}}\bigg)\,.
\end{align*}
\end{theorem}

 Note that when $\alpha \in(1/2,1)$, it holds that $\sum {\gamma_n}=\infty$  and $\sum {\gamma^2_n}<\infty$ for any given $\gamma>0$. Since the convergence rate is not enough to control the remaining error terms under 
$\alpha<\theta<3\alpha-1$ in \cite{Cardot2017} with $\bX=\I_d$, we impose $\theta \in (\alpha, 2\alpha)$ because   $2\alpha>3\alpha-1$ when $\alpha \in(1/2,1)$.  We also conjecture that $\theta=2\alpha$ is the optimal rate for  the $L^4$ convergence and defer its proof to future work.

Directly applying the Borel--Cantelli lemma, we then establish the  almost sure consistency of the SGD estimator as well as the averaged estimator.
\begin{proposition} \label{pr1}
Under Assumptions \ref{as1}-\ref{as4},  
 the SGD estimator (\ref{RM1}) and the  ASGD estimator  (\ref{sq0})satisfy
\[
\lim_{n\rightarrow \infty}\|\bbeta_n-\bbeta\|=0, ~~a.s.\,,~~\lim_{n\rightarrow \infty}\|\bar{\bbeta}_n-\bbeta\|=0, ~~a.s.
\]
respectively.
\end{proposition}
 
\subsection{Weak convergence to a Gaussian process }
 Define the random operators $A_{\bbeta}$ and $B_{\bbeta}$ as 
\begin{align*}
A_{\bbeta}=&\frac{1}{\|Y-\bX^{T}\bbeta\|}\bigg(\mI_{\mathcal H}- \frac{(Y-\bX^{T}\bbeta) \otimes (Y-\bX^{T}\bbeta)}{\|Y-\bX^{T}\bbeta\|^2}\bigg)\,,\\
B_{\bbeta}=& \frac{ \{(Y-\bX\bbeta)\otimes (Y-\bX\bbeta)\}}{\|Y-\bX\bbeta\|^2}\,,
\end{align*}
respectively, and denote $A_0=\mE\{A_{\bbeta}\}$ and $B_0=\mE\{B_{\bbeta}\}$,
where $\mI_{\mathcal H}$ is the identity operator in $\mathcal H$ and $a\otimes b(h)=\langle a, h\rangle b$ can be understood as a linear operator
from $\mathcal H$ to $\mathcal H$ with any $a, b, h\in {\mathcal H}$.  Note that $\|a\otimes b\|$ is an operator norm for the linear operator.

The following  theorem gives the asymptotic normality of the ASGD estimator $\bar{\bbeta}_n$, providing a foundation for statistical inference about 
 $\bbeta$.
\begin{theorem}\label{th3}
Under Assumptions \ref{as1}-\ref{as4}, it follows that
\[
\sqrt{n}(\bar{\bbeta}_n-\bbeta)\xrightarrow{{\mathcal L}} N(0, \Sigma^{-1}\otimes A_{0}^{-1}B_{0}A_{0}^{-1})\,,
\]
where $\xrightarrow{{\mathcal L}}$ denotes the weak convergence.
\end{theorem}

Theorem \ref{th3} establishes the limiting distribution for the averaged  estimator for the slope function $\bbeta$. It provides theoretical support for hypothesis testing as well as  constructing point-wise confidence intervals for $\bbeta$ by the plug-in approach in an offline setting. 
In the multivariate response situation, the averaged estimator $\bar \bbeta_n$ shares the same asymptotic distribution as the classical estimator in \cite{Bai1990}. However, the literature has not yet explored the asymptotic distribution in Hilbert spaces.

\section{Online bootstrap inference for $\beta(t)$ }\label{sec:03}
To apply Theorem \ref{th3} for inference, it is necessary to estimate the asymptotic variance from the data. However, its computation invariably demands substantial computational resources and may become unstable or even infeasible when handling massive or streaming data. Therefore, we introduce an online bootstrap inference procedure by resampling the residuals to approximate the limiting distribution of $\sqrt{n}(\bar{\bbeta}_n-\bbeta)$ given in Theorem \ref{th3}. 
 
Specifically, we consider the wild-type bootstrap based on the Rademacher weight \citep{Canay2021}. Let $W_1, \ldots ,W_n$ be independent and identically distributed random samples of the Rademacher random variable $W$ with  probability mass $\P(W=1)=\P(W=-1)=1/2$. Define the function
\[
G_{\bbeta}(h)=\mE\{\|W(Y-\bX^{T}\bbeta)-\bX^{T}h\|-\|Y-\bX^{T}\bbeta\|\}\,,
\]
and   the gradient of $G_{\bbeta}(h)$ can be calculated as
\[
\nabla G_{\bbeta}(h)=: -\mE\frac{\bX\{W(Y-\bX^{T}\bbeta)-\bX^{T}h\}}{\|W(Y-\bX^{T}\bbeta)-\bX^{T}h\|}\,.
\]
Obviously, $\nabla G_{\bbeta}(h)$ has a unique minimizer at $h=0$ due to  $ -\mE_{X}\frac{X{W(Y-X^{T} \boldsymbol{\beta} )}}{||W(Y-X^{T} \boldsymbol{\beta} )||}=0.$  In practice, the true slope function $\bbeta$ is typically unknown {\it a priori}, thus we can use the ASGD estimator $\bar{\bbeta}_n$ as a plug-in estimator. 
 It is feasible to simultaneously estimate  $\bbeta$ and $h$ by running  two consecutive ASGD algorithms at each recursive iteration. 
 
 As a result,
at the $n$th iteration with a new observed data point $(\bX_{n+1}, Y_{n+1})$, we first update the ASGD estimate for $\bbeta$ based on equations (\ref{RM1}) and  (\ref{sq0}).
Next, we use bootstrap to obtain a large number of replicates for the residual. More specifically, given a large number of $B$ perturbations (say $B=500$), for $b=1, \ldots, B$, we carry out the following steps: 
 \begin{align}
 \varepsilon_{n+1}^{b}=&W_{n+1}^b(Y_{n+1}- \bX_{n+1}^{T} \bar{\bbeta}_{n}), \label{eqq31} \\
\Upsilon_{n+1}^{b}=&\Upsilon_n^{b}+\gamma_n \frac{\bX_{n+1}(\varepsilon_{n+1}^{b}- \bX_{n+1}^{T}\Upsilon_n^{b} )} {\|\varepsilon_{n+1}^{b}-\bX_{n+1}^{T}\Upsilon_n^{b} \|}\,,  \label{eqq32}  \\
	{\bar \Upsilon}_{n+1}^b=&\bar{\Upsilon}_n^b+\frac{1}{n+1}(\Upsilon_{n+1}^b-\bar{\Upsilon}_n^b)\,. \label{eqq33}
\end{align}
This double recursive algorithm finally generates an averaged estimate $\bar{\bbeta}_{n+1}$ and $B$ bootstrap error samples ${\bar \Upsilon}_{n+1}^1, \ldots, {\bar \Upsilon}_{n+1}^B$, which allows us to estimate the limiting  distribution of $\sqrt{n}(\bar{\bbeta}_n-\bbeta)$. This, in turn, facilitates online inference about $\bbeta$, enabling the assessment of estimation efficiency and the construction of confidence intervals or hypothesis testing regarding model parameters.

The following theorem validates this double recursive algorithm via bootstrapping. 

\begin{theorem} \label{th4}
Under Assumptions \ref{as1}-\ref{as4}, for any non-random element $u \in {\mathcal H}^d$, we have
\[
\sup_{t \in \mR}\bigg|\P(\langle u, {\bar \bbeta}_n-\bbeta \rangle \leq t)- \P^{*}( \langle u, {\bar \Upsilon}_{n}^b) \rangle \leq t)\bigg| \rightarrow 0
\] 
in probability,
where $\P^{*}$ stands for the conditional probability  given obvservations.
\end{theorem}
Based on the above  asymptotic normality result, the point-wise $(1-\tau)100\%$  confidence interval for $\beta_j(t)$ at a given location $t$ and $j=1\ldots,d$
can be constructed as
\begin{align}\label{SCI1}
{\mathcal C}^{I}_{n,j}(t)\coloneqq [\bar{\beta}_{n,j}(t)- n^{-1/2}q^B_{1-\tau/2,j}(t),~ \bar{\beta}_{n,j}(t)-
n^{-1/2}q^B_{\tau/2,j}(t) ] 
\end{align}
where $q^B_{\tau/2,j}(t)$ and $q^B_{1-\tau/2,j}(t)$ denote the lower and upper $\tau/2$th percentiles of $\sqrt{n}{\bar \Upsilon}_{n,j}^1(t), \ldots, \sqrt{n}{\bar \Upsilon}_{n,j}^B(t)$.  In Algorithm \ref{al1}, we describe the detailed steps to generate the bootstrap percentile-based point-wise confidence intervals ${\mathcal C}^{I}_{n,j}(t)$ of  $\bbeta(t)$. Note that the bootstrap loop for $b=1, \ldots B$ can be computed in parallel.

 \RestyleAlgo{ruled}
 \begin{algorithm}[t]\label{al1}
 	\caption{Online algorithm for updating the point-wise confidence interval of $\bbeta(t)$ at the confidence level $1-\tau$.}
 	\SetAlgoLined
 	\KwData{Data $\bX_1, \ldots, \bX_{n+1}$; the confidence level $1-\tau$; the number of bootstrap iterations $B$; the initial value $\hat{\bbeta}_0$=$\bar{\bbeta}_0$, where $\bar{\bbeta}_0$ can be a bounded random vector.  }
 	\KwResult{Updated point-wise confidence intervals $[\beta_{n+1,j}^{-}(t),\beta_{n+1,j}^{+}(t)]$ of $\beta_{j}(t)$, for $j=1,\ldots,d$.}
 	\BlankLine
 	
 	\For{{$k\leftarrow 1$ \KwTo $n$}}{
 		
		 	 		Compute $\bar{\bbeta}_{k+1}$ by 
 	 		
 	 		 $$
 	 			\bbeta_{k+1}=\bbeta_k+\gamma_k \frac{\bX_{k+1}(Y_{k+1}- \bX_{k+1}^{T}\bbeta_k )} {\|Y_{k+1}-\bX_{k+1}^{T}\bbeta_k \|}\,,\quad 
 	 			\bar{\bbeta}_{k+1}=\bar{\bbeta}_k+\frac{1}{k+1}(\bbeta_{k+1}-\bar{\bbeta}_k)\,. $$
 		
 		\For{{$b\leftarrow 1$ \KwTo $B$}}{
 			
 			Generate a Rademacher random sample $W_{k+1}^{(b)}$, and compute $ \varepsilon_{k+1}^{b}$ and $\Upsilon_{k+1}^b$ by (\ref{eqq31}) and (\ref{eqq32}).

 			
 			Update ${\bar \Upsilon}_{k+1}^b$ by 
 			
 			
 			$${\bar \Upsilon}_{k+1}^b=\bar{\Upsilon}_k^b+\frac{1}{k+1}(\Upsilon_{k+1}^b-\bar{\Upsilon}_k^b)\,.$$
 		}
 	}

 	
 	\For{{$j\leftarrow 1$ \KwTo $d$}}{
	 	Find the $\tau/2$th and $(1-\tau/2)$th sample quantiles of $\{\sqrt{n+1}{\bar \Upsilon}_{n+1, j}^{(b)}(t)\}_{b=1}^B$, denoted by $q^{B}_{\tau/2,j}$ and $q^{B}_{1-\tau/2,j}$.
		
 		Get $\beta_{n+1,j}^{-}(t)=\bar{\beta}_{n+1,j}(t)-(n+1)^{-1/2}q^{B}_{1-\tau/2,j}(t)$.
 		
 		Get $\beta_{n+1,j}^{+}(t)=\bar{\beta}_{n+1,j}(t)-(n+1)^{-1/2}q^{B}_{\tau/2,j}(t)$.
 	}
 \end{algorithm}

Alternatively, we can also estimate the confidence interval based on the bootstrap variance, as given by 
\begin{align}\label{SCI2}
{\mathcal C}^{II}_{n,j}(t)\coloneqq [\bar{\beta}_{n,j}(t)- z_{1-\tau/2}\sqrt{\hat{\sigma}_{j}^2(t)/n}, \bar{\beta}_{n,j}+
z_{1-\tau/2}\sqrt{\hat{\sigma}_{j}^2(t)/n}] 
\end{align}
where $\hat{\sigma}_{j}^2(t)$ is the sample variance of $\sqrt{n}{\bar \Upsilon}_{n,j}^1(t), \ldots, \sqrt{n}{\bar \Upsilon}_{n,j}^B(t)$, and  $z_{1-\tau/2}$ is the $(1-\tau/2)$th percentile of the standard normal distribution.

\section{Spline interpolation}\label{sec:04}
In practice,  the sample functional responses \{$Y_i(t)\}_{\i=1}^n$ are observed on a common discrete grid ${\bf t} = (t_1, \ldots , t_m )$ in $[0,1]$, where ${\bf t}$ can be a collection of different locations from all individuals.
As a result, no data are observed for $t\in \mathcal{T}/{\bf t}$.  To estimate the entire slope function, we  first consider estimating $\bbeta$ at locations  
$\{t_1, \ldots , t_m \}$  by minimizing the following $L_2$ norm-based loss function using discrete data points, 
\begin{align} \label{eq11}
{\tilde L}_n(\bbeta)=\sum_{i=1}^n \sqrt{\sum_{j=1}^m\bigg\{Y_j(t_j)-\bX_i^{T}\bbeta(t_j)\bigg\}^2}\,.
\end{align}
Let $(\hat{\beta}(t_1), \ldots, \hat{\beta}(t_m))$ denote the minimizer to (\ref{eq11}). 
One can then impose smoothness on $\hat{\bbeta}$ by regressing $\hat{\bbeta}(t_j)$ against $t_j$ using a nonparametric kernel or spline smoothing approach \citep{Rice1991, Hall2006}.
We adopt the $r$-th order spline interpolation approach by solving $\hat{\bbeta}^{r}$ from 
\[
\hat{\bbeta}^{r}=\argmin_{g\in W_2^r} \int \{g^{(r)}(t)\}^2 dt,   
\]
subject to
$$g(t_{l})=\hat{\bbeta}(t_{l}),\quad l=1,\ldots, m.$$
Here, $W_2^r$ denotes the $r$-th order Sobolev space, which is also a Hilbert space. Under the common design, \cite{Cai2011} showed that the minimax rate is of  order $m^{-2r}+n^{-1}$, which is jointly determined by the sampling frequency $m$ and the number of curves $n$. Then we give the convergence rate for $\hat{\bbeta}^{r}$ as follows.
\begin{theorem} \label{th5}
Under Assumptions (\ref{as1}) --  (\ref{as4}), suppose that $\max\|t_{l}-t_{l-1}\|\leq C m^{-1}$ holds,   it follows that
\[
\lim_{D\rightarrow \infty}\P(\|\hat{\bbeta}^{r}-\bbeta\|_2^2\geq D(m^{-2r}+n^{-1}))\rightarrow 0,
\]
as $n\rightarrow \infty$, where $\|\cdot\|_2$ denotes the $L_2$ norm.
\end{theorem}
Theorem \ref{th5} indicates that, when the functional data are observed on a relatively dense grid, such that $m \gg n^{1/(2r)}$, the sampling frequency $m$ does not affect the rate of convergence. This rate is of order $\mathcal{O}(1/n)$, which is solely determined by the sample size $n$.

\section{Numerical studies}\label{sec:05}
To evaluate the finite-sample performance of our proposed online geometric median-based estimation, we assume $\bX=(X_1, X_2, X_3)$ follows a multivariate normal distribution with $\mE(X_i)=0$, $\Var(X_i)=0.5\times 2^{i-1}$, and $\Cor(X_i,X_j)=0.5^{|i-j|}, i,j=1,2,3$. We then generate the functional response $Y(t)$ from model (\ref{functional_model}) by taking $\beta_1(t)=2t^2$,
$\beta_2(t)=\cos(3\pi t/2+\pi/2)$, and $\beta_3(t)=\sin(\pi t/2)+\sqrt{2} (3\pi t/2)$, and the residual process as 
$U(t)=\sum_{l=1}^2\xi_{l}\phi_{l}(t)+\varepsilon_i(t)$, where  $\varepsilon_i(t_k)$ follows a normal distribution $N(0, 0.5)$. We take the basis $\phi_{l}(t)$ as $-\cos\{\pi(t-0.5)\}$ and $\sin\{(t-0.5)\}$ for $l=1,2$, and generate $(\xi_{1},\xi_{2}) $  from either a bivariate normal distribution and a bivariate $t$-distribution (3 degrees of freedom) with mean zero and covariance matrix $0.5\I_2$. The functional responses are observed at $m=50$ locations equally spaced on the interval $[0,1]$.
To construct massive datasets, we set the sample size $n$ as $10000$, $20000$, or $40000$.

To implement our online algorithm, we use the loss function (\ref{eq11}) and evaluate 
different settings for the step size function $\gamma_n=\gamma n^{-\alpha}$ by  fixing $\alpha=0.75$ and varying 
$\gamma \in \{1,1.5,2,3,4,6,10,20\}$. 
Under 1000 simulation replications, we compute the root mean integrated squared error (RMISE) of our online estimator, which is defined as 
\[
{\rm RMISE}(k)=\sqrt{m^{-1}\sum_{j=1}^m\big(\bar\beta_{n,k}(t_j)-\beta_k(t_j)\big)^2 }\,,
\]
for $k=1, 2,3.$

As shown in Table \ref{tab1},  the proposed online estimator demonstrates robust performance when the tuning parameter $\gamma$ is selected from the range $[2,20]$, with a larger sample size yielding a reduced RMISE, This, in turn, further numerically substantiates our theoretical convergence results. 
Upon comparing results under different distributional assumptions for the residual process, it is observed that the RMISE is relatively larger for the heavy-tailed $t$ distribution compared to the normal distribution, aligning with our expectations.

\begin{table}[htp]
\caption{Root mean integrated squared errors for the estimates of $\bbeta=(\beta_1,\beta_2,\beta_3)$ based on the proposed online approach across different values of the step size $ \gamma$ and the sample size $n$. Mean ($\times 10^{-2}$) and standard deviation ($\times 10^{-2}$, in parentheses) are given.}\label{tab1}
\begin{tabular*}{\textwidth}{@{\extracolsep\fill}lcccccc}
\toprule%
& \multicolumn{3}{@{}c@{}}{$(\xi_{1},\xi_{2}) \sim $ bivariate normal distribution} & \multicolumn{3}{@{}c@{}}{$(\xi_{1},\xi_{2}) \sim $  bivariate $t$-distribution} \\\cmidrule{2-4}\cmidrule{5-7}%
$\gamma$ & $\beta_1$ & $\beta_2$ & $\beta_3$ & $\beta_1$ & $\beta_2$ & $\beta_3$ \\
\midrule
\multicolumn{7}{@{}c@{}}{$n=10000$}\\
1.0   & 1.55(0.49) & 1.11(0.33) & 0.64(0.15)  & 1.70(0.57) & 1.41(0.49) & 0.88(0.30) \\
1.5   & 1.28(0.30) & 1.10(0.26) & 0.64(0.16)   & 1.74(0.58) & 1.37(0.46) & 0.88(0.30)\\
2.0   & 1.30(0.33) & 1.03(0.25) & 0.64(0.15)& 1.77(0.62) & 1.40(0.48) & 0.89(0.31) \cr
  3.0   & 1.28(0.31) & 1.04(0.24) & 0.64(0.15) &1.73(0.59) & 1.39(0.45) & 0.90(0.31) \cr
  4.0   & 1.31(0.32) & 1.04(0.24) & 0.63(0.15) &  1.72(0.58) & 1.41(0.49) & 0.89(0.30)  \cr
   6.0   & 1.30(0.29) & 1.04(0.25) & 0.65(0.16) &1.74(0.60) & 1.41(0.48) & 0.90(0.30)  \cr
   10    & 1.31(0.30) & 1.06(0.26) & 0.65(0.16) & 1.78(0.61) & 1.41(0.46) & 0.90(0.32)  \cr
   20    & 1.35(0.34) & 1.09(0.27) & 0.67(0.17) &1.78(0.61) & 1.42(0.47) &  0.91(0.31)   \cr
   \multicolumn{7}{@{}c@{}}{$n=20000$}\\
   1.0   & 0.87(0.21) & 0.68(0.16) & 0.43(0.10)  & 1.17(0.40) & 0.91(0.31) & 0.58(0.19)  \\
1.5   & 0.87(0.20) & 0.68(0.17) & 0.43(0.11)   & 1.19(0.41) & 0.94(0.32) & 0.58(0.19)  \\
2.0   & 0.87(0.22) & 0.69(0.17) & 0.43(0.10) &  1.17(0.39) & 0.93(0.31) & 0.58(0.19)  \cr
  3.0   & 0.86(0.21) & 0.69(0.17) & 0.43(0.11)  &1.18(0.40) & 0.93(0.31) & 0.58(0.19) \cr
  4.0   & 0.85(0.21) & 0.68(0.16) & 0.43(0.10) &  1.16(0.39) & 0.93(0.31) & 0.59(0.20)  \cr
   6.0   & 0.85(0.19) & 0.68(0.17) & 0.44(0.10) &1.15(0.38) & 0.91(0.30) & 0.57(0.19)  \cr
   10    & 0.86(0.21) & 0.68(0.16) & 0.43(0.10) & 1.16(0.40) & 0.92(0.31) & 0.59(0.20)  \cr
   20    & 0.89(0.22) & 0.69(0.16) & 0.45(0.10) &1.20(0.42) & 0.93(0.31) & 0.60(0.21)   \cr
      \multicolumn{7}{@{}c@{}}{$n=40000$}\\
   1.0   &  0.60(0.14) & 0.48(0.12) & 0.30(0.07) &  0.86(0.52) & 0.68(0.37) & 0.41(0.15)   \\
1.5   & 0.59(0.14) & 0.47(0.11) & 0.30(0.07)   & 0.86(0.84) & 0.67(0.59) & 0.41(0.18)  \\
2.0   & 0.60(0.15) & 0.48(0.12) & 0.30(0.07) &  0.81(0.26) & 0.65(0.23) & 0.39(0.13)   \cr
  3.0   & 0.60(0.14) & 0.47(0.11) & 0.29(0.07)  &0.82(0.29) & 0.65(0.22) & 0.40(0.14) \cr
  4.0   & 0.59(0.14) & 0.47(0.12) & 0.30(0.07)&   0.79(0.27) & 0.64(0.23) & 0.40(0.14)  \cr
   6.0   &0.59(0.14) & 0.47(0.12) & 0.30(0.08)&0.81(0.29) & 0.64(0.21) & 0.40(0.14)  \cr
   10    & 0.60(0.14) & 0.47(0.12) & 0.30(0.08)& 0.81(0.28) & 0.64(0.21) & 0.40(0.14)   \cr
   20    &  0.60(0.15) & 0.48(0.12) & 0.30(0.07) &0.81(0.28) & 0.64(0.21) & 0.41(0.14)  \cr   
\botrule
\end{tabular*}
\end{table}

To assess the efficiency of our proposed online estimation algorithm, we also implemented the following three competing methods for comparison: (a) The offline geometric median-based algorithm for solving equation (\ref{eq11}) based on the full ``static'' data, which can be implemented using the R package ``MNM,'' and (b) The offline pointwise median-based estimator of \cite{Liu2020}, and (c) The offline least squares estimator for the FSRM. We compared the results of these offline methods to those of our online geometric median-based approach with a tuning parameter $\gamma$ set to 3. Notably, the offline geometric median-based algorithm (method (a)) serves as an oracle benchmark for our proposed online approach. This is because the offline method utilizes the full dataset, while our online method relies only on the most recent data point at each iteration.

\begin{figure}[hbt]
	\centering
	(a) $(\xi_{1},\xi_{2}) \sim $ bivariate normal distribution\\
		\includegraphics[width=13cm, height=6.5 cm]{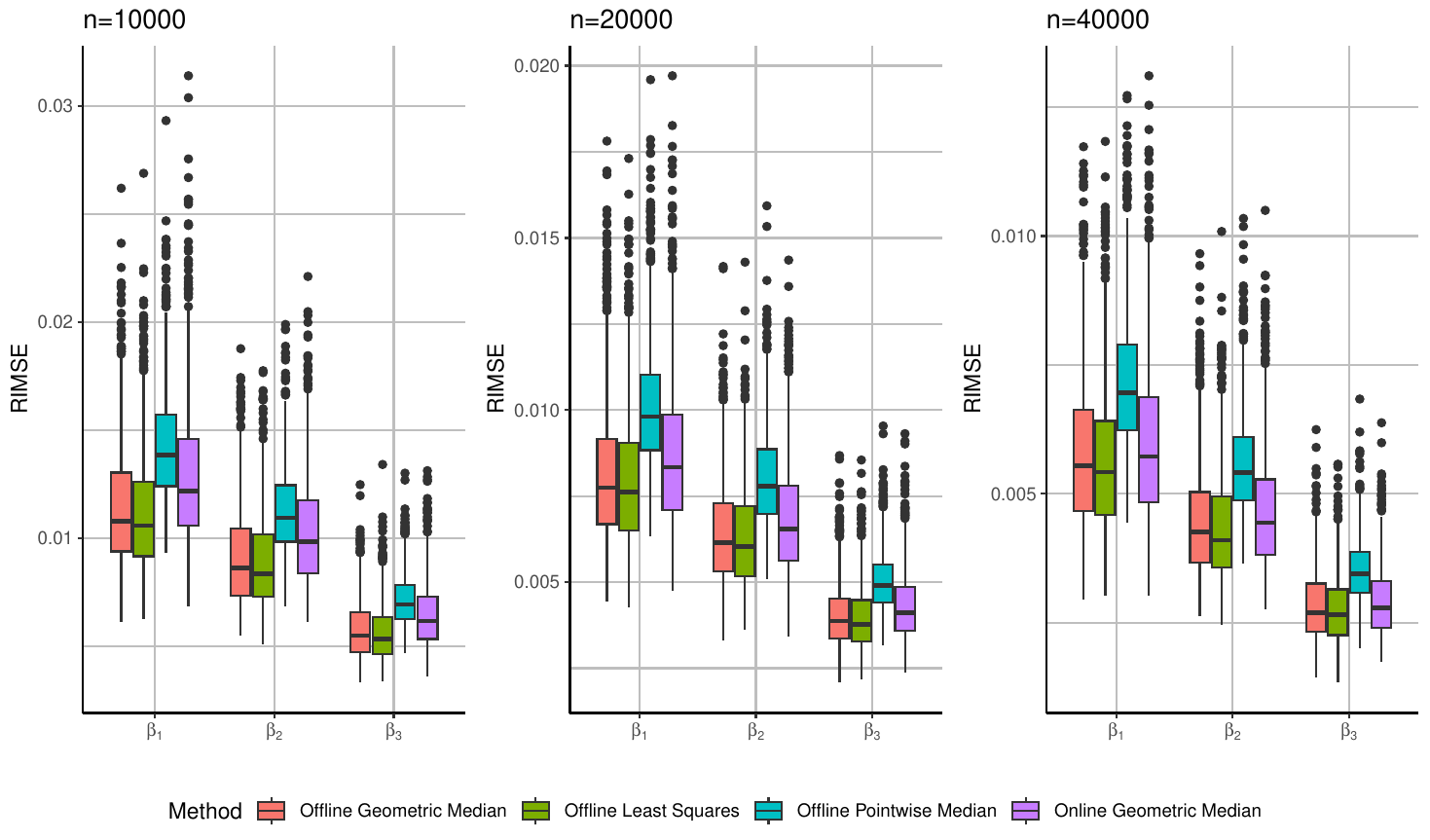}\\
		
		\vspace{0.1in}
		
		(b) $(\xi_{1},\xi_{2}) \sim $ bivariate $t$-distribution\\
		\includegraphics[width=13cm, height=6.5 cm]{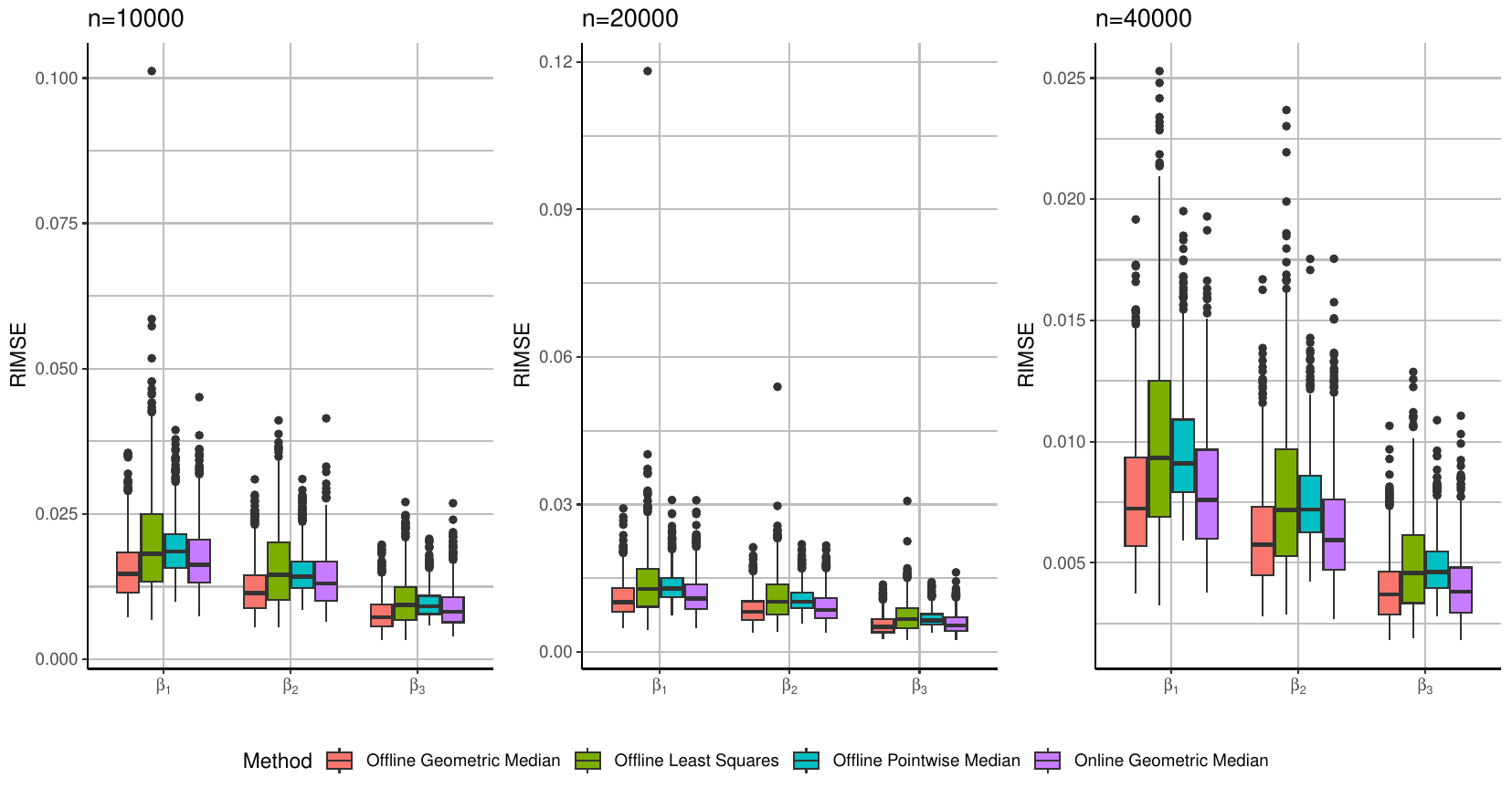}
	
		\caption{Boxplots of the root mean integrated squared errors (RMSIE) for the  offline geometric median-based estimator, the offline pointwise median-based estimator,  the offline least square estimator, and the proposed proposed  online geometric median-based estimator  according to 1000 simulation replications under (a) $(\xi_{1},\xi_{2}) \sim $ bivariate normal distribution and (b) $(\xi_{1},\xi_{2}) \sim $ bivariate $t$-distribution.}
	\label{Fig1}
\end{figure}

Figure \ref{Fig1} demonstrates that the offline least squares estimator achieves the smallest average root mean integrated squared error (RMISE) under the Gaussian setting. Interestingly, both the offline and online geometric median methods perform competitively with the offline least squares estimator in this setting. As anticipated, under the heavy-tailed setting, the geometric median-based methods outperform the others in terms of RMISE. Across all scenarios and settings, our online geometric median-based algorithm exhibits nearly identical performance to its offline counterpart. This simulation study highlights the robustness of the offline geometric median-based estimation and shows that the use of the recursive algorithm does not lead to a significant loss in efficiency. Moreover, as the sample size $n$ increases, the online algorithm offers significant advantages in terms of storage and computational time. Notably, when running 100 simulation replications with $n=10000$ on a MacBook with an Intel Core i5 processor, our online method required only 0.16 seconds, while the offline method took 305 seconds to complete.

Next, we evaluate the performance of the point-wise confidence intervals derived from our proposed online bootstrap method. Using the same settings from the above simulation with a normal residual process and a sample size of $n=10000$, Figure \ref{Fig3} displays the empirical coverage probabilities for the two proposed point-wise confidence intervals, the bootstrap percentile-based ${\mathcal C}_{n,j}^{I}$ and bootstrap variance-based ${\mathcal C}_{n,j}^{II}$, $j=1,2,3$, at the 90\% and 95\% confidence levels. The results demonstrate that both ${\mathcal C}_{n,j}^{I}$ and ${\mathcal C}_{n,j}^{II}$ are capable of achieving the target coverage probabilities across various location points.

\begin{figure}[hbt]
	\centerline{
		\includegraphics[width=13cm, height=6cm]{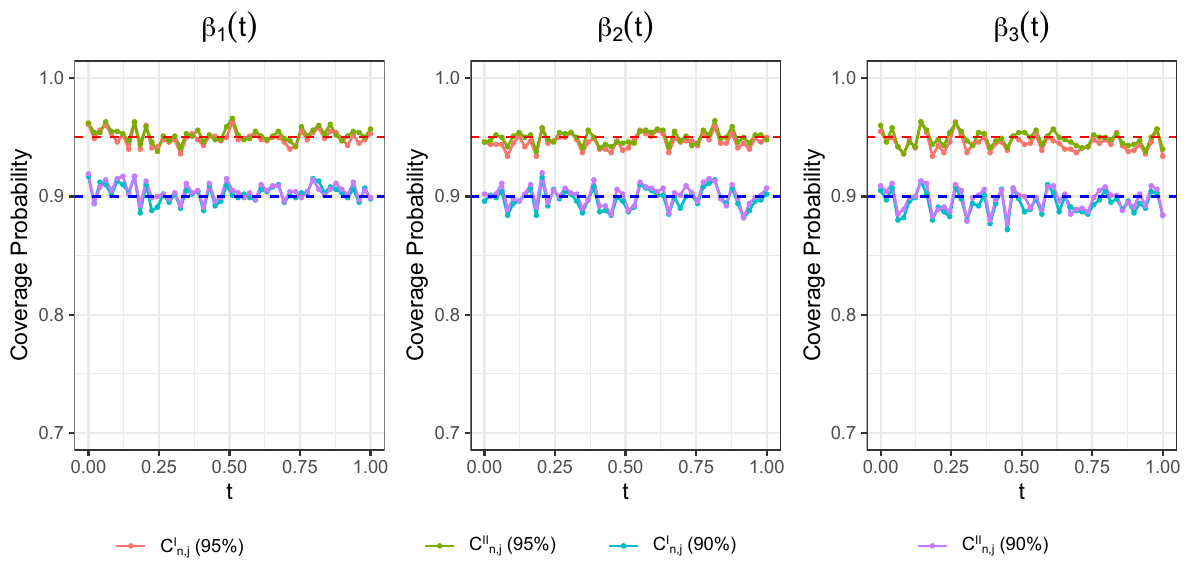}
	}
		\caption{Point-wise confidence intervals based on our proposed online bootstrap procedure. ${\mathcal C}_{n,j}^{I}$ and ${\mathcal C}_{n,j}^{II}$  denote the bootstrap percentile- and variance-based confidence intervals, respectively. The number in the parentheses represents the  confidence level.}
	\label{Fig3}
\end{figure}

\section{Real data application}\label{sec:06}
We applied our proposed online geometric median-based approach for the FSRM (\ref{functional_model}) to analyze the Beijing multi-site air-quality data.\cite{Zhang2017}. This dataset,  publicly available on the UCI machine learning repository \url{https://archive.ics.uci.edu}, comprises hourly air pollution readings from a total of 12 air quality monitoring stations, collected between March 1, 2013, and February 28, 2017. The air pollution readings encompass hourly concentrations of  PM$_{{2.5}}$, O$_{3}$,  SO$_2$, NO$_2$ and CO, along with other environmental indicators such as temperature (TEMP), atmospheric pressure (PRES), dew point temperature (DEWP), and wind speed (WSPM). In this application, we were particularly interested in $Y(t)=$ time-varying PM$_{2.5}$, a functional measure of fine inhalable particles, with its levels indicating the potential risk of health problems associated with inhalation.
The covariates $\bX$ were taken as the  daily average  values of O$_{3}$,  SO$_2$, NO$_2$ , CO  and TEMP, PRES, DEWP, WSPM.

In our analysis, we excluded hours with missing values for any of the selected variables and pooled the data from the 12 air quality monitoring sites. To mitigate site-specific effects, we centered and standardized the data at each site individually. We also standardized the time variable 
$t$ to fall within the interval $[0,1]$.

Figure \ref{Fig4} demonstrates the trajectories of estimated slope functions corresponding to each air pollution covariate at time points $t=0, 0.304, 0.652, 1$.  It shows that the estimation stabilizes as more data accumulate. Evidently,  a sample size of 2500 is sufficient to yield stable estimates based on our proposed online estimation method.

 \begin{figure}[t!]
	\centerline{
		\includegraphics[width=13cm, height=9cm]{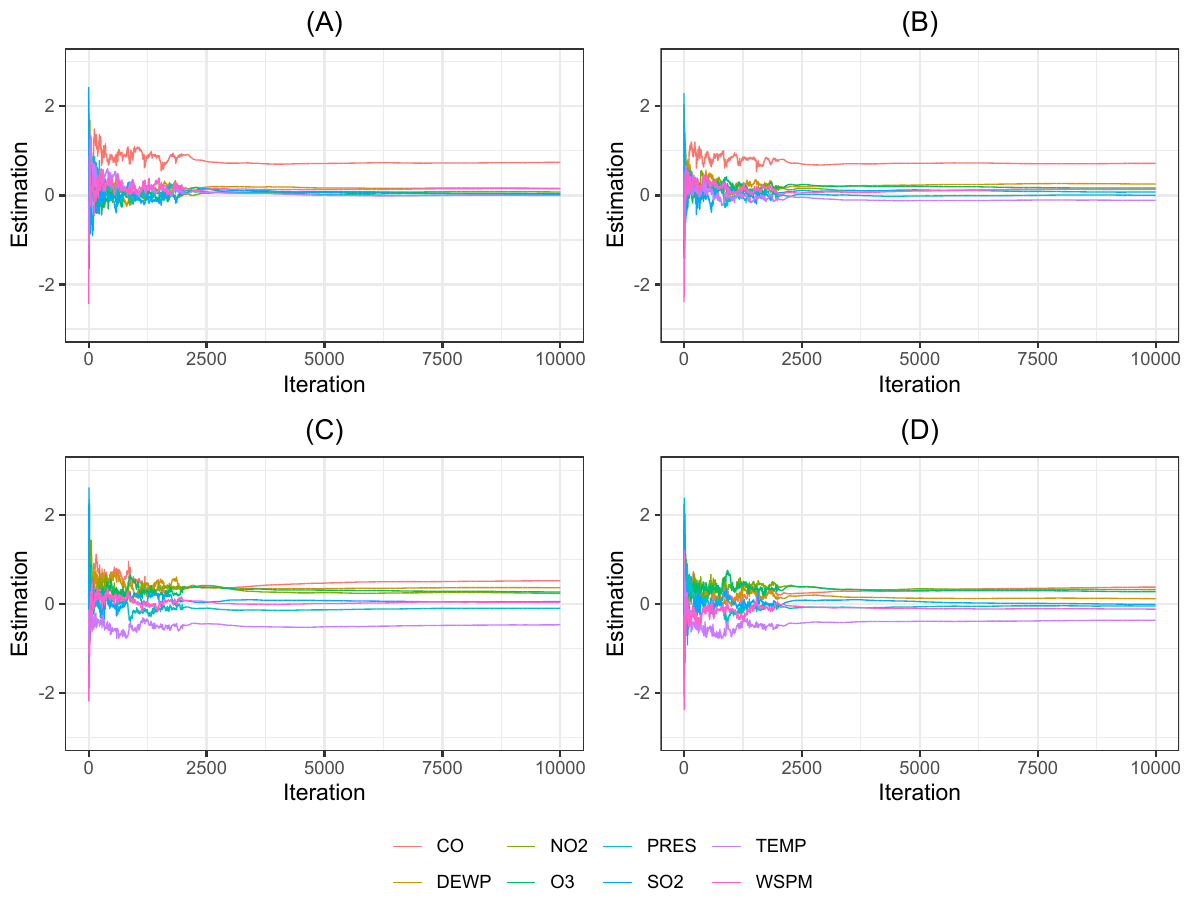}
	}
		\caption{Convergence trajectories of estimated slope functions corresponding to each
air pollution covariate at  time points: (A) $t=0.1$,  (B)  $t=0.304$, (C) $t=0.652$, and (D)  $t=1$.}
	\label{Fig4}
\end{figure}

\begin{figure}[t!]
	\centerline{
		\includegraphics[width=13cm, height=7cm]{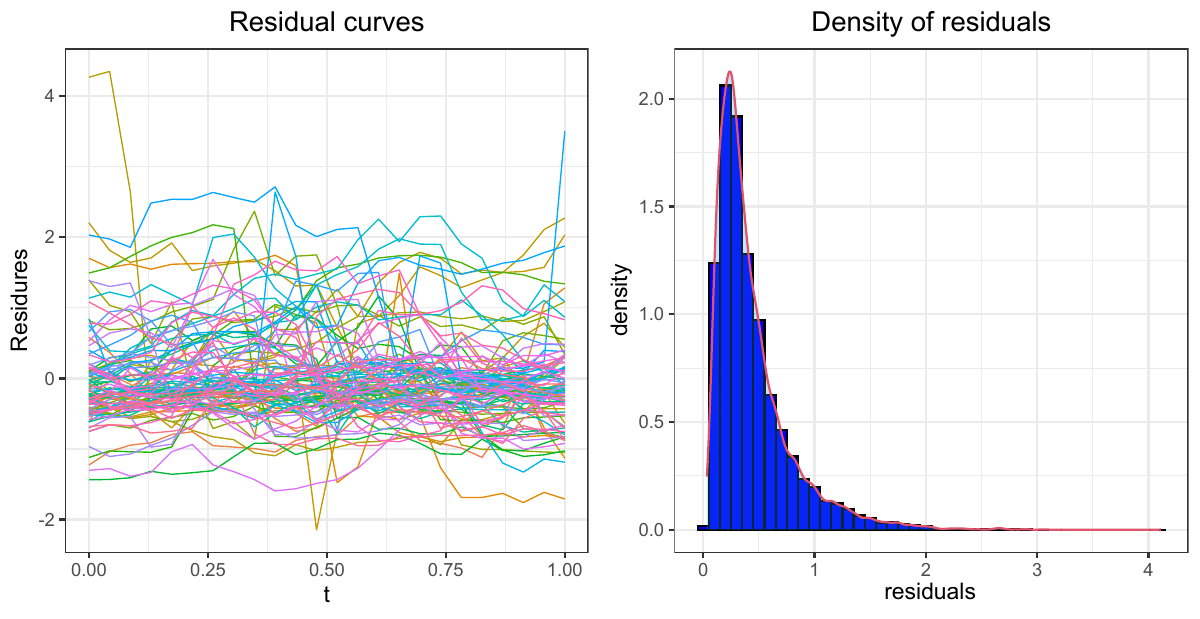}
	}
		\caption{The left panel plots residual curves at 24 grid points using spline interpolation from 100 randomly selected. The right panel estimated density of the integrated absolute residuals obtained from our proposed algorithms.}
	\label{Fig5}
\end{figure}

Figure \ref{Fig5} displays the residual curves at 24 grid time points, interpolated using the R function \texttt{spline}, along with the estimated density of the integrated absolute residuals. The residuals predominantly fluctuate around the zero line, and the density exhibits a slight right skewness with potential heavy tails. This pattern suggests that a robust geometric median-based functional regression, as proposed, may be more appropriate for analyzing the data than a mean-based regression approach.

Figure \ref{Fig6} presents the estimated slope functions and the 90\% point-wise confidence intervals for the eight air pollution covariates. It is evident that CO, NO$_2$,  O$_3$, and DEWP  exert significantly positive effects on the hourly PM$_{{2.5}}$. Notably, the impact of CO on PM$_{{2.5}}$
 appears to diminish over time, whereas the effects of NO$_2$ and O$_3$ seem to intensify.  The confidence intervals for the slope functions of  SO$2$ and PRES include zero at most time points, suggesting their negligible influence on  PM$_{{2.5}}$ levels.

\begin{figure}[t!]
	\centerline{
		\includegraphics[width=13cm, height=7cm]{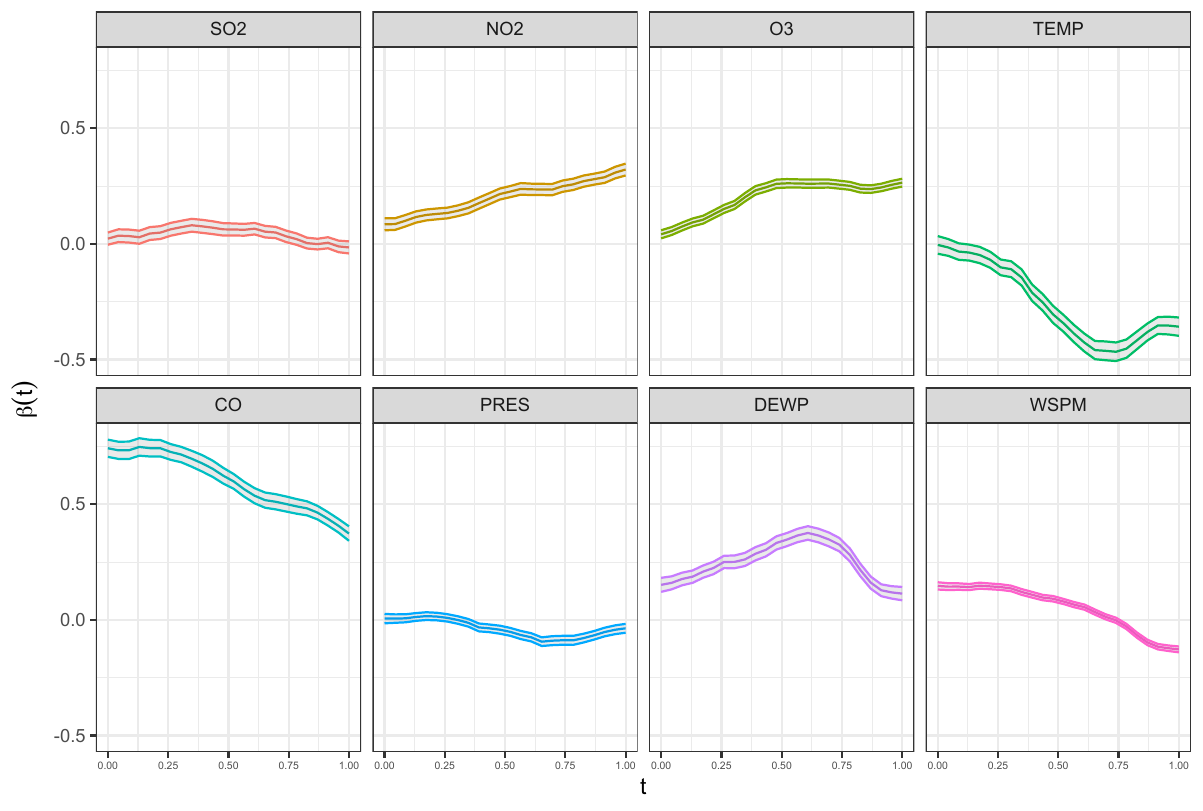}
	}
		\caption{ Estimated slope functions for the eight air pollution covariates as well as the  90\% point-wise confidence
intervals, based on the proposed online boostrap method via geometric median. The point-wise confidence band is shown in light grey.}
	\label{Fig6}
\end{figure}

\section{Conclusion}
We have introduced a novel function-on-scalar regression methodology using the geometric median, complemented by online estimation and bootstrap inference techniques suited for massive or streaming datasets. Our theoretical and numerical analyses validate the effectiveness of these proposed online methods. Notably, our geometric median-based regression demonstrates superior performance with heavy-tailed distributions compared to traditional mean-based approaches. Moreover, even when dealing with data from a Gaussian distribution, our method maintains efficiency comparable to that of mean-based techniques. Supporting theoretical evidence for these claims is detailed in \cite{Cheng2023}.

One interesting direction for future work involves extending our method to accommodate high-dimensional covariates $\bX$, as suggested by \cite{Barber2017} and \cite{Fan2017}. Employing a geometric median-based approach in such settings could potentially enhance robustness in feature selection and inference. Additionally, refining the current methodology by developing an optimal interpolation technique for discretely sampled functional data represents another valuable direction for improvement.

\begin{appendices}

\section{Some useful Propositions and Lemmas}

In the following, we provide  some  propositions and lemmas that are necessary for the proofs of Theorem \ref{th1} and \ref{th3}.
Denote the gradient function and the corresponding hessian function as
\[
\Phi(h)=-\mE \left\{\frac{ \bX(Y-\bX^{T}h)}{\|Y-\bX^{T}h\|}\right\}\,,~ \Gamma_{h}=\mE\bigg( \bX\bX^{T} \otimes A_h \bigg )\,,
\]
where 
\[
A_h=\frac{1}{\|Y-\bX^{T}h\|}\bigg(\mI_{\H}- \frac{(Y-\bX^{T}h) \otimes (Y-\bX^{T}h)}{\|Y-\bX^{T}h\|^2}\bigg)\,,
\]
and note that when $h=\bbeta$, $\Gamma_{\bbeta}=\Sigma \otimes A_0$ with $A_0=\mE(A_\bbeta)$, and $\I_{\H}$ is the identity operator in $\H$.  

The following proposition gives an important result for $\Gamma_\bbeta$.
\begin{proposition} \label{pr2}
Let $\{\lambda_{i,\Sigma}\}_{i=1}^d$ be the eigenvalues of $\Sigma=\mE(\bX \bX^{T})$, we have $1/C \leq \lambda_{i,\Sigma} \leq C$. Moreover, for any 
 $\phi \in {\mathcal H}^d$, there is a positive constant $C_{\bbeta}$ such that,
 \[
C_{\bbeta}\|\phi\|^2 \leq \langle \phi, \Gamma_{\bbeta} \phi\rangle  \leq C\|\phi\|^2\,,
\]
which indicates $C_{\bbeta}\leq \lambda_{i,\bbeta}\leq C$, with $\lambda_{i,\bbeta}$ being  the eigenvalues of $\Gamma_\bbeta$.
\begin{proof}
Let $(\lambda_{i, \Sigma})$ be the eigenvalues $\Sigma=\mE(\bX \bX^{T})=\{\sigma_{ij}\}_{i,j=1}^d$,  and $\Gamma_{\beta}$ can be viewed as a $d\times d$ matrix
with elements being covariance operators. For any $\phi=(\phi_1, \ldots, \phi_d)^T\in {\mathcal H}^d$,  it holds that
\[
\langle \phi, \Gamma_{\bbeta}\phi \rangle =\sum_{i=1}^d\sum_{j=1}^d\sigma_{ij} \langle \phi_i, A_0 \phi_j\rangle=
\sum_{k}\lambda_{k, A_{0}} \sum_{i=1}^d\sum_{j=1}^d\sigma_{ij}\langle e_k,  \phi_i\rangle \langle e_k,  \phi_j \rangle\,,
\]
where $e_k$ is the eigenfunction corresponding to the $k$th eigenvalue $\lambda_{k,A_0}$ of $A_{0}=\mE(A_{\bbeta})$, and the last equality is implied by  the eigenvalue decomposition 
for $A_{0}$.
Together with   Assumption \ref{as1}, $\Sigma$ is positive definite and $1/C \leq \lambda_{i,\Sigma}\leq C$, then it holds that
\[
\sum_{k}\lambda_{k, A_{0}} \sum_{i=1}^d\sum_{j=1}^d\sigma_{ij}\langle e_k,  \phi_i\rangle \langle e_k,  \phi_j \rangle=\sum_{k}\lambda_{k, A_{0}}
\langle \ba_k, \Sigma \ba_k \rangle\,,
\]
where $\ba_k=(a_{k,1}, \ldots, a_{k,d})$ and $a_{k,j}=\langle e_k,  h_j \rangle$. Hence, we have
\[
1/C \sum_{k}\lambda_{k, A_0} \|\langle e_k,  h \rangle\|^2\leq \sum_{k}\lambda_{k, A_0} \langle \ba_k, \Sigma \ba_k \rangle  \leq C \sum_{k}\lambda_{k, A_0} \|\langle e_k,  h \rangle\|^2\,.
\]
By the assumption that $\mE\|U-h\|^{-2}\leq C$ and  Proposition 2.1 in \cite{Cardot2017},  there exist two positive constants $c_{\bbeta}$ and  $C$,  such that for all $\psi \in {\mathcal H}$, 
\[
c_{\bbeta}\|\psi\|^2 \leq \langle \psi, A_{0}\psi \rangle \leq C \|\psi\|^2\,.
\]
By taking    $C_{\bbeta}=c_{\bbeta}/C$, it holds that
\[
C_{\bbeta}\|h\|^2 \leq \langle h, \Gamma_{\bbeta}h \rangle  \leq C\|h\|^2\,,
\]
which completes the proof.
\end{proof}
\end{proposition}

\begin{proposition}\label{pr3}
Under Assumptions \ref{as1} -- \ref{as3}, there exists a constant $C$, such that  
\[
\|\Phi(\bbeta_k)-\Gamma_{\bbeta}(\bbeta_k-\bbeta)\| \leq C\|\bbeta_k-\bbeta\|^2\,,
\]
for all $k\geq 1$.
\end{proposition}
\begin{proof}
Since $\Phi(\bbeta)=0$,  applying Taylor’s expansion, we have
\[
\Phi(\bbeta_k)=\int_0^1\Gamma_{\bbeta+t(\bbeta_k-\bbeta)}(\bbeta_n-\bbeta) \d t\,.
\]
 Let $r_k=\Phi(\bbeta_k)-\Gamma_{\bbeta}(\bbeta_k-\bbeta)=\int_0^1(\Gamma_{\bbeta+t(\bbeta_k-\bbeta)}-\Gamma_{\bbeta})(\bbeta_k-\bbeta)\d t$. Following   \citet{Cardot2017}, we define the function $\varphi_{h, h^{'}}(\cdot)$  from $[0,1]\rightarrow  {\mathcal H}^d$ as
\[
\varphi_{h,h^{'}}(t):=\Gamma_{\bbeta+th}(h^{'})\,,
\]
for all $h, h^{'} \in {\mathcal H}^d$. Then, 
 we have 
\begin{align*}
\|\Phi(\bbeta+h)-\Gamma_{\bbeta}(h)\|=&\|\int_0^1(\Gamma_{\bbeta+th}-\Gamma_{\bbeta})(h)\d t\| \leq \int_0^1
\|(\Gamma_{\bbeta+th}(h)-\Gamma_{\bbeta}(h))\|\d t\\
=& \int_0^1 \|\varphi_{h,h}(t)-\varphi_{h,h}(0)\| \d t \leq \sup_{t\in [0,1]}\|\varphi^{'}_{h,h}(t)\|\,.
\end{align*}
Next we  give an upper bound for $\sup_{t\in [0,1]}\|\varphi^{'}_{h,h}(t)\|$.  Let
\[
W_h(t)=\frac{1}{\|Y-\bX^{T}\bbeta-t\bX^{T}h\|}, ~V_{h, h^{'}}(t)=h^{'}-\frac{\langle Y-\bX^{T}\bbeta-t\bX^{T}h, h^{'}\rangle(Y-\bX^{T}\bbeta-t\bX^{T}h)}{\|Y-\bX^{T}\bbeta-t\bX^{T}h\|^2}\,,
\]
then we have $\varphi_{h,h^{'}}^{'}(t)=\mE\{\bX\bX^{T}\otimes (W^{'}_h(t)V_{h,h^{'}}(t)+W_h(t)V^{'}_{h,h^{'}}(t))\}$. By algebra, it holds that
\[
\|V_{h,h^{'}}(t)\|\leq 2\|h^{'}\|,~~~ \|W^{'}_{h}(t)\|\leq \frac{\|\bX^{T}h\|}{\|Y-\bX^{T}\bbeta-t\bX^{T}h\|}\,,
\]
and
\[
\|V^{'}_{h,h^{'}}(t)\| \leq \frac{4\|\bX^{T}h\|\|h^{'}\|}{\|Y-\bX^{T}\bbeta-t\bX^{T}h\|}\,.
\]
Since $\bX$ is bounded almost surely, it holds that $\|\varphi_{h,h^{'}}^{'}(t)\| \leq 6\|h\|\|h^{'}\|\mE\{\|U-t\bX^{T}h\|^{-2}\}\leq C\|h\|\|h^{'}\|$. Thus, we complete the proof by choosing $h=\bbeta_k-\bbeta$.
\end{proof}

\begin{lemma} \label{lem1}
Under Assumptions \ref{as1} -- \ref{as4},  there exist  some constants $C$ and  $C^{'}$ such that for all $n\geq 1$, 
\[
\mE\{\|\bbeta_n-\bbeta\|^2\} \leq Ce^{-C^{'}n^{1-\alpha}}+\frac{C}{n^{\alpha}}+C\sup_{n/2-1\leq k \leq n}\mE\{\|\bbeta_{k}-\bbeta\|^4\}\,.
\]

\end{lemma}
\begin{proof}
  Denote that $\Gamma_\bbeta=\mE(\bX\bX^{T} \otimes A_{\bbeta} )$,  then we decompose the online SGD estimator as 
\begin{align}\label{lm1eq1}
\bbeta_{n+1}-\bbeta=\bbeta_n-\bbeta-\gamma_n \Gamma_\bbeta(\bbeta_n-\bbeta)+\gamma_n \xi_{n+1}-\gamma_n r_n\,,
\end{align}
where
\[
\xi_{n+1}=\frac{\bX_{n+1}(Y_{n+1}- \bX_{n+1}^{T}\bbeta_n )} {\|Y_{n+1}- \bX_{n+1}^{T}\bbeta_n \|}+\Phi(\bbeta_n)\,.
\]
Denote $H=\H^d$ and define 
$
\kappa_k=\I_{H} -\gamma_k \Gamma_{\bbeta}, ~~ \nu_n=\prod_{k=1}^n\kappa_k$. Note that $\nu_0=\I_{H}\,,
$
we then rewrite equation (\ref{lm1eq1}) as
\[
\bbeta_{n+1}-\bbeta=\nu_n(\bbeta_1-\bbeta)+\nu_n M_{n+1}-\nu_n R_{n+1}\,,
\]
where 
$
R_{n+1}=\sum_{k=1}^n\gamma_k \nu_k^{-1}r_k$,  and $M_{n+1}=\sum_{k=1}^n\gamma_k \nu_k^{-1}\xi_{k+1}$.
Applying the inequality of arithmetic means, we have
\[
\mE\|\bbeta_n-\bbeta\|^2 \leq  3\mE(\|\nu_{n-1}(\bbeta_1-\bbeta)\|^2)+3\mE(\|\nu_{n-1}M_n\|^2)+3\mE(\|\nu_{n-1}R_n\|^2)\,.
\] 
In the following steps, we subsequently bound $\mE(\|\nu_{n-1}(\bbeta_1-\bbeta)\|^2)$, $\mE(\|\nu_{n-1}M_n\|^2)$ and $\mE(\|\nu_{n-1}R_n\|^2)$. Using the results of   Proposition \ref{pr2}, it holds that $C_{\bbeta} \leq \lambda_{i,\bbeta}\leq C$.  Next, due to fact
$
\|\nu_{n-1}\| \leq  C_1\exp\{-C_{\bbeta}\sum_{k=1}^{n-1}\gamma_k\}
$
for some constant $C_1$,   it holds that
\begin{align*}
\mE(\|\nu_{n-1}(\bbeta_1-\bbeta)\|^2)\leq C_1^2\bigg\{\exp\bigg(-C_{\bbeta}\sum_{k=1}^{n-1}\gamma_k\bigg)\bigg\}^2\mE\|\bbeta_1-\bbeta\|^2\leq C\exp(-C^{'}n^{1-\alpha})\,.
\end{align*}
To bound $\mE\{\|\nu_{n-1}M_n\|^2\}$, we define $U_{n+1}:=- \frac{ \bX_{n+1}(Y_{n+1}-\bX_{n+1}^{T}\bbeta_n)}{\|Y_{n+1}-\bX_{n+1}^{T}\bbeta_n\|}$\,. Then $\xi_n$ can be represented as $\xi_n=\Phi(\bbeta_n)-U_{n+1}$, and $\mathcal F_n$ is a $\sigma-$algebra  generated by $\sigma(\bbeta_1,\ldots,\bbeta_n)=\sigma(Z_1,\ldots,Z_n)$, where $Z_i=(\bX_i, Y_i)$. Apparently, for all integers $n\geq 1$,
\[
  \mE(U_{n+1}|\mathcal F_n)=\Phi(\bbeta_n).
\] 
This leads to $\{\xi_n\}_{n=1}$ being a sequence of martingale differences. Due to the orthogonality of the martingale differences, we have $\mE\langle \xi_{k}, \xi_{k^{'}}\rangle=0$ for any $k\neq k^{'}$. Hence it follows that
\[
\mE\{\|\nu_{n-1}M_n\|^2\}=\sum_{k=1}^{n-1}\gamma_k^2\mE\{\|\nu_{n-1}\nu_k^{-1}\xi_{k+1}\|^2\}\,.
\]
Since $\|\xi_{k+1}\|^2\leq 4\|\bX_{n+1}\|^2\leq C$ almost surely,   we have $\mE\{\|\nu_{n-1}M_n\|^2\}\leq C\sum_{k=1}^{n-1}\gamma_k^2\|\nu_{n-1}\nu_k^{-1}\|^2$.  Using  similar calculations in \cite{Cardot2017}, we have $\sum_{k=1}^{n-1}\gamma_k^2\|\nu_{n-1}\nu_k^{-1}\|^2\leq Cn^{-\alpha}$. Then it holds that 
$
\mE\|\nu_{n-1}M_n\|^2\leq  Cn^{-\alpha}
$ for all $n\geq 1$.

To bound $\mE(\|\nu_{n-1}R_n\|^2)$,  we write $\bbeta_{n+1}-\bbeta=\bbeta_n-\bbeta-\gamma_n \Phi(\bbeta_n)+\gamma_n\xi_{n+1}$,  then it holds that
\begin{align*}
\mE(\|\bbeta_{n+1}-\bbeta\|^2|\F_n)=&\|\bbeta_n-\bbeta\|^2+\gamma_n^2\{\|\Phi(\bbeta_n)\|^2+\mE(\|\xi_{n+1}\|^2|\F_n)\}\\
&~~~~~~-2\gamma_n \langle \bbeta_n-\bbeta, \Phi(\bbeta_n)\rangle\,.
\end{align*}
As $\|\Phi(\bbeta_n)\|^2+\mE(\|\xi_{n+1}\|^2|\F_n)\leq 2\|\bX_{n+1}\|^2$ and $\langle \bbeta_n-\bbeta, \Phi(\bbeta_n)\rangle \geq 0$, by taking the expectation of the two sides, we have $\mE(\|\bbeta_{n+1}-\bbeta\|^2)\leq \mE(\|\bbeta_{n}-\bbeta\|^2)+C\gamma_n^2$, which leads to 
$\mE(\|\bbeta_{n}-\bbeta\|^2)\leq C\sum_{k=1}^{\infty}\gamma_n^2\lesssim C$. Similarly, since $\|\bbeta_{n+1}-\bbeta\|^2\leq \|\bbeta_n-\bbeta\|^2+C\gamma_n^2+2\gamma_n
\langle \xi_{n+1}, \bbeta_n-\bbeta \rangle$,   it holds that
\[
\mE(\|\bbeta_{n+1}-\bbeta\|^4|\F_n)\leq \|\bbeta_n-\bbeta\|^4+C\gamma_n^4+4\gamma_n^2 \mE(\|\langle \xi_{n+1}, \bbeta_n-\bbeta \rangle\|^2|\F_n)+C\gamma_n^2\|\bbeta_n-\bbeta\|^2\,.
\]
Together with $\|\xi_{n+1}\|\leq 2\|\bX_{n+1}\|^2$,   $\mE(\|\bbeta_{n+1}-\bbeta\|^4)\leq \mE(\|\bbeta_{n}-\bbeta\|^4)+C\gamma_n^2$, we can show that $\mE(\|\bbeta_{n}-\bbeta\|^4)\leq C$ for any $n\geq 1$.  Because    $\|r_{n}\|\leq C\|\bbeta_n -\bbeta\|^2$ by   Proposition \ref{pr3},  it holds that
\begin{align*}
\mE(\|\nu_{n-1}R_n\|^2)\leq &\mE\left\{\left(\sum_{k=1}^{n-1}\gamma_k\|\nu_{n-1}\nu_k^{-1}\|_2\|r_{k}\|\right)^2\right\}\\
&~~~~~~~~~~\leq C\mE\left\{\left(\sum_{k=1}^{n-1}\gamma_k\|\nu_{n-1}\nu_k^{-1}\|_2 \|\bbeta_n -\bbeta\|^2\right)^2\right\}\,.
\end{align*}
By Lemma 5.5 in \cite{Cardot2017}, it thus holds that
\begin{align*}
\mE(\|\nu_{n-1}R_n\|^2)\leq &~C \left(\sum_{k=1}^{n-1}\gamma_k\|\nu_{n-1}\nu_k^{-1}\|_2\mE\{\|{\bbeta}_k-\bbeta\|^4\}^{1/2}\right)^2 \\
&~~~~~~~\leq  C\left(\sum_{k=1}^{E(n/2)} \gamma_k\|\nu_{n-1}\nu_k^{-1}\|_2\right)^2\\
&~~~~~~+ C \sup_{E(n/2)-1 \leq k\leq n} \mE\{\|\bbeta_n -\bbeta\|^4\}\left(\sum_{k=E(n/2)+1}^{n} \gamma_k\|\nu_{n-1}\nu_k^{-1}\|_2\right)^2\,,
\end{align*}
where $E(\cdot)$ denotes the integer function.  It follows that
\begin{align*}
\mE\|\nu_{n-1} R_{n}\|^2 \leq C\left(\sum_{k=1}^{E(n/2)}\gamma_k e^{-C^{'}\sum_{j=k+1}^{n-1}\gamma_k}\right)^2
+C\sup_{E(n/2)-1 \leq k\leq n-1} \mE\|\bbeta_{k}-\bbeta\|^4\,.
\end{align*}
Therefore, we can obtain 
\[
\mE\|\nu_{n-1} R_{n}\|^2 \leq Ce^{-C^{'}n^{1-\alpha}}+\frac{C}{n^{\alpha}}+C\sup_{E(n/2)-1\leq k \leq n}\mE\{\|\bbeta_{k}-\bbeta\|^4\}\,,
\]
which completes the proof
\end{proof}
\begin{lemma} \label{lem2}
Under Assumptions \ref{as1} -- \ref{as3},   there exists a  positive constant $C$ which depends on $q$ such that
\begin{align*}
\mE\{\|\bbeta_{n+1}-\bbeta\|^4\} \leq \bigg(1-\frac{1}{n^\frac{1+(q-1)\alpha}{q}}\bigg)^2\mE\{\|\bbeta_n-\bbeta\|^4\}+\frac{C}{n^{3\alpha}}
+\frac{C}{n^{2\alpha}}\mE\{\|\bbeta_n-\bbeta\|^2\}\,,
\end{align*}
 for    all $n\geq n_{0}$ dependent on $\alpha$, some $n_{0}$, and $q\geq 1$.
\begin{proof}
First of all, we write $\|\bbeta_{n+1}-\bbeta\|^2$ as
\begin{align*}
\|\bbeta_{n+1}-\bbeta\|^2=&\|\bbeta_n-\bbeta-\gamma_n \Phi(\bbeta_n)\|^2+\gamma_n^2\|\xi_{n+1}\|^2
+2\gamma_n\langle \bbeta_n-\bbeta-\gamma_n \Phi(\bbeta_n), \xi_{n+1} \rangle \\
\leq &\|\bbeta_n-\bbeta-\gamma_n \Phi(\bbeta_n)\|^2+2\gamma_n\langle \bbeta_n-\bbeta, \xi_{n+1} \rangle
+C\gamma_n^2\,.
\end{align*}
 Applying Cauchy-Schwarz’s inequality and the law of iterated expectation   with the fact $\mE\{\langle \xi_{n+1}, \bbeta_n-\bbeta-\gamma_n \Phi(\bbeta_n)\rangle \|\bbeta_n-\bbeta\||\F_n\}=0$, we have
\[
\mE\{\|\bbeta_{n+1}-\bbeta\|^4\}\leq \mE\{\|\bbeta_n-\bbeta-\gamma_n \Phi(\bbeta_n)\|^4\}+
C\gamma_n^2\mE\{\|\bbeta_{n+1}-\bbeta\|^2\}+C\gamma_n^4\,.
\]
Since $\gamma_n^4=o(1/n^{3\alpha})$,   there exists a positive constant $C$ and $n_0$ such that 
\[
\mE\{\|\bbeta_{n+1}-\bbeta\|^4\}\leq \mE\{\|\bbeta_n-\bbeta-\gamma_n \Phi(\bbeta_n)\|^4\}+
C\frac{1}{n^{2\alpha}}\mE\{\|\bbeta_{n+1}-\bbeta\|^2\}+C\frac{1}{n^{3\alpha}}\,,
\]
for all $n\geq n_0$.
Then we aim at bounding $\mE\{\|\bbeta_n-\bbeta-\gamma_n \Phi(\bbeta_n)\|^4\}$. Define the set of the sequence of events 
for some $q\geq 1$ as,
\[
S_{n,q}:=\bigg\{\omega: \|\bbeta_n(\omega)-\bbeta\|\leq Mn^{\frac{1-\alpha}{q}}\bigg\}\,.
\]
If $\|\bbeta_n-\bbeta\|\leq 1$, then $\|\bbeta_n\|\leq \|\bbeta\|+1$. By Proposition \ref{pr2}, it holds that
$\langle \Phi(\bbeta_n)-\Phi(\bbeta), \bbeta_n-\bbeta\rangle \geq C_{\bbeta}\|\bbeta_n-\bbeta\|^2$. As a result, there exists $n_0$, 
for all $n\geq n_0$, $C_{\bbeta}\geq M n^{-\frac{1-\alpha}{q}}$,    we have
$\langle \Phi(\bbeta_n), \bbeta_n-\bbeta\rangle \geq M n^{-\frac{1-\alpha}{q}}$. If $\|\bbeta_n-\bbeta\|\geq 1$, then
\begin{align*}
\langle \Phi(\bbeta_n), \bbeta_n-\bbeta\rangle =&\int_0^1\langle \Gamma_{\bbeta+t(\bbeta_n-\bbeta)}(\bbeta_n-\bbeta)
, \bbeta_n-\bbeta\rangle \d t\\
&\geq \int_0^{1/\|\bbeta_n-\bbeta\|}\langle \Gamma_{\bbeta+t(\bbeta_n-\bbeta)}(\bbeta_n-\bbeta), \bbeta_n-\bbeta\rangle \d t\\
&\geq \int_0^{1/\|\bbeta_n-\bbeta\|}  C_{\beta} \|\bbeta_n-\bbeta\|^2\d t \geq \frac{C_{\bbeta}}{Mn^{(1-\alpha)/q}}\|\bbeta_n-\bbeta\|^2.
\end{align*}
It thus holds that   $\|\bbeta_n-\bbeta-\gamma_n \Phi(\bbeta_n)\|^2 \leq \|\bbeta_n-\bbeta\|^2+\gamma_n^2\|\bbeta_n-\bbeta\|^2-\frac{2C_{\bbeta}}{Mn^{(1-\alpha)/q}}\frac{\gamma}{n^{\alpha}}\|\bbeta_n-\bbeta\|^2$. Moreover, by choosing
$M=2\gamma C_{\bbeta}$,  we have $\|\bbeta_n-\bbeta-\gamma_n \Phi(\bbeta_n)\|^2\leq (1-n^{-\frac{1+(q-1)\alpha}{q}})\|\bbeta_n-\bbeta\|^2$. Since $\|\bbeta_n-\bbeta\|\leq \|\bbeta_{n-1}-\bbeta\|+\gamma_{n-1}\leq \|\bbeta_1-\bbeta\|+\sum_{i=k}^{n-1}\gamma_k$ for  all $n\geq n_0$. Then it follows that   $\|\bbeta_n-\bbeta\|\leq Cn^{1-\alpha}$ for all $n\geq 1$, and thus 
\begin{align*}
\mE\{\|\bbeta_n-\bbeta-\gamma_n \Phi(\bbeta_n)\|^4\}=&~\mE\{\|\bbeta_n-\bbeta-\gamma_n \Phi(\bbeta_n)\|^4\I_{S_{n,q}}\}\\
&~~~~~+\mE\{\|\bbeta_n-\bbeta-\gamma_n \Phi(\bbeta_n)\|^4\I_{S_{n,q}^c}\}\\
&\leq \bigg(1-n^{-\frac{1+(q-1)\alpha}{q}} \bigg)^2\mE\{\|\bbeta_n-\bbeta\|^4\}+C\frac{n^{4-4\alpha}}{n^{2p(1-\alpha)/q}}\,,
\end{align*}
where the last term is implied by $\mE\{\|\bbeta_n-\bbeta-\gamma_n \Phi(\bbeta_n)\|^4\I_{S_{n,q}^c}\}\leq Cn^{4-4\alpha}\P(\|\bbeta_n-\bbeta\|\geq Mn^{(1-\alpha)/q})$.   Using Markov's inequality with the fact $\mE(\|\bbeta_{n}-\bbeta\|^p)\leq C$ and mathematical induction,  we have 
\[
\P(\|\bbeta_n-\bbeta\|\geq Mn^{\frac{1-\alpha}{q}})\leq \frac{\mE(\|\bbeta_n-\bbeta\|^{2p})}{\{Mn^{\frac{1-\alpha}{q}}\}^{2p}}\leq C\frac{1}{n^{2p\frac{1-\alpha}{q}}}\,.
\]
Then by choosing $p>q\frac{4-\alpha}{2(1-\alpha)}$ it holds that 
\begin{align*}
\mE\{\|\bbeta_n-\bbeta-\gamma_n \Phi(\bbeta_n)\|^4\}\leq \bigg(1-n^{-\frac{1+(q-1)\alpha}{q}} \bigg)^2\mE\{\|\bbeta_n-\bbeta\|^4\}+C\frac{1}{n^{3\alpha}}\,,
\end{align*}
which completes the proof.
\end{proof}
\end{lemma}

\section{Main Proofs}
\subsection{Proof of Theorem \ref{th1}}

\begin{proof}[Proof of Theorem \ref{th1}]
Let   $\theta \in(\alpha, 2\alpha)$ and $q>\frac{1-\alpha}{2\alpha-\theta}$,  we have $3\alpha-\theta>\frac{1-(q-1)\alpha}{q}$.
There exists  $n_0 \geq 1$, which depends on $\alpha , \theta$, such that  for all $n\geq n_0$,
\[
\bigg(1-n^{-\frac{1+(q-1)\alpha}{q}} \bigg)^2\bigg(\frac{n+1}{n}\bigg)^{\theta}+\frac{C}{(n+1)^{3\alpha-\theta}}=1-2
n^{-\frac{1+(q-1)\alpha}{q}}+o(n^{-\frac{1+(q-1)\alpha}{q}})\leq 1\,.
\]
We show by mathematical  induction that there are positive constants $C_{\alpha}$ and $C_{\theta}$  ($C_{\alpha}\leq C_{\theta}\leq 2C_{\alpha}$) such that 
\[
\mE\{\|\bbeta_n-\bbeta\|^2\}\leq \frac{C_{\alpha}}{n^{\alpha}}, ~~~\mE\{\|\bbeta_n-\bbeta\|^4\}\leq \frac{C_{\theta}}{n^{\theta}}\,,
\]
for any $n\geq n_0$. Because   $\max\{\mE\{\|\bbeta_n-\bbeta\|^2\}, \mE\{\|\bbeta_n-\bbeta\|^4\}\}\leq C$,   we then choose $C_{\alpha}\geq n_0\{\mE\{\|\bbeta_n-\bbeta\|^2\}\}$ and $C_{\theta}\geq n_0\{\mE\{\|\bbeta_n-\bbeta\|^4\}\}$.  Using  Lemma \ref{lem2} by induction, it holds that with $C_{\theta}\geq 1$,
\begin{align*}
\mE\{\|\bbeta_{n+1}-\bbeta\|^4\}\leq &~ \bigg(1-n^{-\frac{1+(q-1)\alpha}{q}} \bigg)^2\frac{C_{\theta}}{n^{\theta}}+\frac{CC_{\alpha}}{n^{3\alpha}}\\
&~~~~\leq \bigg(1-n^{-\frac{1+(q-1)\alpha}{q}} \bigg)^2\bigg(\frac{n+1}{n}\bigg)^{\theta} \frac{C_{\theta}}{(n+1)^{\theta}}\\
&~~~~~~~~+C_{\alpha}\bigg(\frac{n+1}{n}\bigg)^{3\alpha}\frac{1}{(n+1)^{3\alpha-\theta}}\frac{C_{\theta}}{(n+1)^{\theta}}
\leq \frac{C_{\theta}}{(n+1)^{\theta}}\,.
\end{align*}
With   Lemma \ref{lem1}, it also holds that 
\[
\mE\{\|\bbeta_{n+1}-\bbeta\|^2\} \leq \frac{C}{(n+1)^{\alpha}}+C2^{\theta+1}\frac{C_{\alpha}}{(n+1)^{\theta}}\,.
\]
By choosing $C_{\alpha}\geq C$ and $\frac{C2^{\theta+1}}{(n+1)^{\theta-\alpha}}\leq 1$ for all $n\geq n_0$,  we have $\mE\{\|\bbeta_{n+1}-\bbeta\|^2\} \leq \frac{C_{\alpha}}{(n+1)^{\alpha}}$. As a result, we complete the proof by taking $C_{\alpha}\geq \max_{n\leq n_0}\{n_0^{\alpha} \mE\{\|\bbeta_n-\bbeta\|^2\}\}$  and $C_{\theta}\geq \max_{n\leq n_0}\{n_0^{\theta} \mE\{\|\bbeta_n-\bbeta\|^4\}\}$.
\end{proof}

\subsection{Proof of Proposition  \ref{pr1}}
\begin{proof}[Proof of Proposition \ref{pr1}]
In the proof of  Theorem \ref{th1},  we have shown that  $\mE\{\|\bbeta_n-\bbeta\|^2\}=\mO(n^{-\alpha})$ and $\mE\{\|\bbeta_n-\bbeta\|^4\}=\mO(n^{-\theta})$,  To apply Borel-Cantelli’s Lemma, we only need to show that, there exists some $\delta>0$, 
\begin{align}\label{th2eq1}
\sum_{n\geq 1}\P\bigg(\|\bbeta_n-\bbeta\|\geq \frac{1}{n^{\delta}}\bigg)< \infty\,.
\end{align}
According to Markov's inequality, it holds that
\[
\sum_{n\geq 1}\P\bigg(\|\bbeta_n-\bbeta\|\geq \frac{1}{n^{\delta}}\bigg)\leq \sum_{n\geq 1}n^{4\delta} \mE(\|\bbeta_n-\bbeta\|^4)
\leq \sum_{n\geq 1}\frac{C}{n^{\theta-4\delta}}\,.
\]
By choosing $1<\theta<2\alpha$ and a constant $\delta$ satisfying $\delta<(\theta-1)/4$, then it holds that $\sum_{n\geq 1}n^{-\theta+4\delta}< \infty$. Therefore, equation (\ref{th2eq1}) holds and  we have
\[
\lim_{n\rightarrow \infty} \|\bbeta_n-\bbeta\|\rightarrow 0~~ a.s\,.
\]
Finally, the almost sure consistency of $\bar{\bbeta}_n$ is obtained by a direct application of Toeplitz’s lemma.
\end{proof}
\subsection{Proof of Theorem \ref{th3}}
\begin{proof}[Proof of Theorem \ref{th3}]
Recall the decomposition 
$
\bbeta_{k+1}=(\I_H-\gamma_k \Gamma_\bbeta)\bbeta_k+\gamma_k\xi_{k+1}-\gamma_kr_{k}\,,
$
then by algebra, it holds that
\[
\Gamma_\bbeta \bbeta_k=\xi_{k+1}-r_k+\frac{1}{\gamma_k}(\bbeta_k-\bbeta_{k+1})\,.
\]
The averaged estimator can be re-written as 
\[
\sqrt{n}\Gamma_\bbeta  (\bar{\bbeta}_n-\bbeta)=\frac{1}{\sqrt{n}}\bigg(\frac{T_1}{\gamma_1}-\frac{T_{n+1}}{\gamma_n}+\sum_{k=2}^nT_k\bigg(\frac{1}{\gamma_k}-\frac{1}{\gamma_{k-1}}\bigg)-\sum_{k=1}^n r_k\bigg)+\frac{1}{\sqrt{n}}\sum_{k=1}^n\xi_{k+1}\,,
\]
where $T_n=\bbeta_n-\bbeta$. As $T_1$ is bounded almost surely,   we have  $n^{-1/2}\gamma_1^{-1}T_1=o(1), a.s$. By Theorem \ref{th1}, $\mE\{\|n^{-1/2}\gamma_n^{-1}T_n\|^2\}=\mO(n^{-1}n^{2\alpha}n^{-\alpha})=\mO(n^{\alpha-1})$,  we have $n^{-1/2}\gamma_n^{-1}T_n=o_p(1)$.
Moreover, as $|\gamma_k^{-1}-\gamma_{k-1}^{-1}|\leq 2\alpha \gamma k^{\alpha-1}$, it holds that
\begin{align*}
\frac{1}{n}\mE\left\|\sum_{k=2}^nT_k\left(\frac{1}{\gamma_k}-\frac{1}{\gamma_{k-1}}\right)\right\|^2 \leq & \frac{1}{n}\bigg(\sum_{k=2}^n |\gamma_k^{-1}-\gamma_{k-1}^{-1}| \sqrt{\mE\|\bbeta_k-\bbeta\|^2}\bigg)^2 \\
&~~~~~~~~~\lesssim \frac{1}{n}4\alpha^2\bigg(\sum_{k=2}^n k^{\alpha/2-1}\bigg)^2=\mO\left(\frac{1}{n^{1-\alpha}}\right)\,.
\end{align*}
Hence, we have $n^{-1/2}\sum_{k=2}^n T_k(\gamma_k^{-1}-\gamma_{k-1}^{-1}) \xrightarrow{P} 0$. Finally, by  choosing $\theta\in (1,2\alpha)$, $n^{-1/2}\sum_{k=1}^nr_k\xrightarrow{P} 0$, it holds that
\begin{align*}
\frac{1}{n}\mE\left\|\sum_{k=1}^n r_k\right\|^2\leq 
& \frac{1}{n}\bigg(\sum_{k=1}^n\sqrt{\mE\|\bbeta_k-\bbeta\|^4}\bigg)^2
\leq \frac{1}{n}\bigg(\sum_{k=1}^n \frac{1}{k^{\theta/2}}\bigg)^2
=\mO\bigg(\frac{1}{n^{\theta-1}}\bigg)\,.
\end{align*}
Let $\Xi_n=\mE(\xi_{n+1}\otimes\xi_{n+1}|\F_n)$, then we have
\[
\Xi_n=\mE\bigg(\frac{\bX_{n+1}\bX_{n+1}^{T}\otimes \{(Y_{n+1}-\bX_{n+1}\bbeta_n)\otimes (Y_{n+1}-\bX_{n+1}\bbeta_n)\}}{\|Y_{n+1}-\bX_{n+1}\bbeta_n\|^2}\bigg| \F_n\bigg)-\Phi(\bbeta_n)\otimes\Phi(\bbeta_n)\,.
\]
By some algebra, it holds that
\[
\|\Phi(\bbeta_n)\|\leq \mE\bigg(\frac{2}{\|Y-\bX^{T}\bbeta\|}\bigg)\|\mE(\bX\bX^{T})\| \|\bbeta_n-\bbeta\|\,.
\]
Under Assumption \ref{as3}, as $\mE\{\|U-h\|^{-1}\}\leq C$,  we have $\|\Phi(\bbeta_n)\|\leq  C\|\bbeta_n-\bbeta\|\rightarrow 0 $ almost surely.  We then define 
\[
B_0=\mE\bigg(\frac{ \{(Y_{n+1}-\bX_{n+1}\bbeta)\otimes (Y_{n+1}-\bX_{n+1}\bbeta)\}}{\|Y_{n+1}-\bX_{n+1}\bbeta\|^2} \bigg)\,.
\]
Based on some  similar computation for $\Phi(\bbeta_n)$, and together with Assumption \ref{as3}, it holds that
\[
\|\Xi_n-\bX_n\bX_n^{T}\otimes B_0\|\leq C\|\bbeta_n-\bbeta\|\rightarrow 0, ~~a.s\,.
\]
To show the asymptotic normality,  we only need to check the conditions of the functional central limit theorem (CLT)  for Hilbert-valued  martingales \citep{Lavrentyev2016}. For any orthonormal basis $e_i, e_j \in {\mathcal H}^d$, it holds that
\[
\bigg|\frac{1}{n}\sum_{k=1}^n\mE\{\langle \xi_k\otimes \xi_k e_i, e_j\rangle|\F_{k-1}\}-\langle (\bD \otimes B_0) e_i, e_j\rangle\bigg|\leq \|\bD\|\frac{1}{n}\sum_{k=1}^n\|\bbeta_k-\bbeta\|
\xrightarrow{P} 0\,,
\]
where $\bD=n^{-1}\sum_{i=1}^n\bX_i\bX_i^{T}$. Similarly, it also can be shown that $|\mE\{\| \xi_k\|^2|\F_{k-1}\}-\tr (\bX_n\bX_n^{T}\otimes B_0)|\leq C\|\bbeta_n-\bbeta\|\xrightarrow{P} 0$. Then we have
\[
\bigg|\frac{1}{n}\sum_{k=1}^n\mE\{\| \xi_k\|^2|\F_{k-1}\}-\tr \{\bD \otimes B_0\}\bigg| \xrightarrow{P} 0\,.
\]
Finally, because of $\max_{k}\|\xi_k\|\leq \max_{k}\|\bX_k\|\leq C$, and $\|\bD-\Sigma \|\rightarrow 0$ almost surely,  
it holds that $\frac{1}{\sqrt{n}}\sum_{k=1}^n\xi_{k+1} \xrightarrow{{\mathcal L}}  N(0, \Sigma\otimes B_0)$ under Assumption \ref{as1}-\ref{as4}. Since  $\Gamma_{\bbeta}$ is positive definite according to Proposition \ref{pr2}, it holds that
\[
\bbeta_n-\bbeta \xrightarrow{{\mathcal L}}  N(0, \Sigma^{-1}\otimes A_{0}^{-1}B_0A_{0}^{-1})\,,
\]
which completes the proof.
\end{proof}

\subsection{Proof of Theorem \ref{th4}}
\begin{lemma} \label{lem3}
Suppose Assumptions \ref{as1}--\ref{as3} are satisfied,  we have for any $n\geq 1$,
$$\|\nabla G_{\bar{\bbeta}_n}( \Upsilon_n)- {\hat \Gamma}_{\bar{\bbeta}_n}( { \Upsilon}_n) \| \leq C\|\Upsilon_n\|^2,$$ where ${\hat \Gamma}_{\bar{\bbeta}_n}=\nabla^2 G_{\bar{\bbeta}_n}(0)$.
\end{lemma}

\begin{proof}
Note that $\nabla G_{\bar{\bbeta}_n}(0)=0$,  applying Taylor’s expansion with integral remainders, we have
\begin{align*}
\nabla G_{\bar{\bbeta}_n}( \Upsilon_n)=&\nabla G_{\bar{\bbeta}_n}(0)+\int_0^1\nabla^2 G_{\bar{\bbeta}_n}(t(\Upsilon_n))\Upsilon_n \d t\\
=&\int_0^1\nabla^2 G_{\bar{\bbeta}_n}(t\Upsilon_n)\Upsilon_n \d t\,.
\end{align*}
It then holds that
\begin{align*}
\|\nabla G_{\bar{\bbeta}_n}( \Upsilon_n)- {\hat \Gamma}_{\bar{\bbeta}_n}({\Upsilon}_n)\|=&\left\|\int_0^1\{\nabla^2 G_{\bar{\bbeta}_n}(0)-\nabla ^2 G_{\bar{\bbeta}_n}(t\Upsilon_n)\} \Upsilon_n \d t \right\|\\
& \leq \int_0^1\left\|\{\nabla^2 G_{\bar{\bbeta}_n}(0)-\nabla ^2 G_{\bar{\bbeta}_n}(t\Upsilon_n)\} \Upsilon_n \right\| \d t \,.
\end{align*}
 According to the proof of Proposition \ref{pr2},  we have
\[
\left\|\{\nabla^2 G_{\bar{\bbeta}_n}(0)-\nabla ^2 G_{\bar{\bbeta}_n}(t\Upsilon_n)\} \Upsilon_n \right\|\leq C\mE\{\|W(Y-\bX^{T}\bar{\bbeta}_n)-t\bX^{T}\Upsilon_n\|^{-2}\|\Upsilon_n\|^2|\F_n\}\leq C\|\Upsilon_n\|^2\,.
\]
the last inequality is implied by Assumption \ref{as3} with $\|W(Y-\bX^{T}\bar{\bbeta}_n)-t\Upsilon_n\|^{-2} \leq \|Y-\bX^{T}\bar{\bbeta}_n-t\Upsilon_n\|^{-2}+\|Y-\bX\bar{\bbeta}_n+t\Upsilon_n\|^{-2}$. 
\end{proof}
\begin{lemma} \label{lem4}
Suppose Assumptions \ref{as1}--\ref{as3} are satisfied,  we have $$\|({\hat \Gamma}_{\bar{\bbeta}_n}-\Gamma_{\bbeta}){\Upsilon}_n\|\leq C \|\bar{\bbeta}_n-\bbeta\|\|\Upsilon_n\|,$$ for any $n\geq 1$.
\end{lemma}

\begin{proof}
Recall that $\Gamma_{\bbeta}$ is defined as 
\[
\Gamma_{\bbeta}= \Sigma \otimes \mE\bigg\{\frac{1}{\|Y-\bX^{T}\bbeta\|}\bigg(\I_{H}-\frac{(Y-\bX^{T}\bbeta) \otimes(Y-\bX^{T}\bbeta)}{\|Y-\bX^{T}\bbeta\|^2}\bigg)\bigg\}\,.
\]
We express $\varphi_{h, h^{'}}$, a function $ [0,1]\rightarrow  \H^d$, as
\begin{align*}
\varphi_{h, h^{'}}(t)=&\Gamma_{\bbeta+th} (h^{'})\\
=&\Sigma\otimes\mE\bigg\{\frac{1}{\|Y-\bX^{T}\bbeta-t \bX^{T}h\|}\bigg(\I_{\H}-\frac{(Y-\bX^{T}\bbeta-t\bX^{T}h)\otimes(Y-\bX^{T}\bbeta-t\bX^{T}h)}{\|Y-\bX^{T}\bbeta-t\bX^{T}h\|^2}\bigg) \bigg\}(h^{'})\,.
\end{align*}
Then we can write $({\hat \Gamma}_{\bar{\bbeta}_n}-\Gamma_{\bbeta}){\Upsilon}_n=\mE\{\varphi_{\bar{\bbeta}_n-\bbeta, \Upsilon_n}(1)-\varphi_{\bar{\bbeta}_n-\bbeta, \Upsilon_n}(0)|\F_n\}$, and it thus follows that
\[
\mE\{\varphi_{\bar{\bbeta}_n-\bbeta, \Upsilon_n}(1)-\varphi_{\bar{\bbeta}_n-\bbeta, \Upsilon_n}(0)|\F_n\}=\int_0^1
\mE\{\varphi^{'}_{\bar{\bbeta}_n-\bbeta}(t)|\F_n\} \d t\,.
\]
It also holds that $\|\int_0^1\mE\{\varphi^{'}_{\bar{\bbeta}_n-\bbeta, \Upsilon_n}(t)|\F_n\} \d t\| \leq  \sup_{t \in[0,1]}\mE\{\|\varphi^{'}_{\bar{\bbeta}_n-\bbeta, \Upsilon_n}(t)\||\F_n\}$. Based on the proof of Proposition \ref{pr3}, we have
\[
\mE\{\|\varphi^{'}_{h, h^{'}}(t)\|\}\leq 6C \|h\|\|h^{'}\|\}\,.
\]
Hence, it follows that \[
\mE\{\|\varphi^{'}_{\bar{\bbeta}_n-\bbeta, \Upsilon_n}(t)\||\F_n\}\leq C\mE\{\|Y-\bX^{T}\bbeta-t\bX^{T}(\bar{\bbeta}_n-\bbeta)\|^{-2}\|\bar{\bbeta}_n-\bbeta\|\|\Upsilon_n\||\F_n\}\leq C \|\bar{\bbeta}_n-\bbeta\|\|\Upsilon_n\|\,,
\]
which completes the proof.
\end{proof}
The following lemma gives the convergence rates of $\Upsilon_n$ in $L_2$ and $L_4$. As its proof   
 is  particularly similar to that of Theorem \ref{th1}, we omit it here.
\begin{lemma} \label{lem5}
Under Assumptions \ref{as1} -- \ref{as4}, for all $ \theta\in (\alpha, 2\alpha)$ and $n\geq 1$, it holds that
\begin{align*}
\mE\{\|\Upsilon_n\|^2\}=\mO\bigg(\frac{1}{n^{\alpha}}\bigg),~~~~
\mE\{\|\Upsilon_n\|^4\}=\mO\bigg(\frac{1}{n^{\theta}}\bigg)\,.
\end{align*}
\end{lemma}

Next, we give the proof of Theorem  \ref{th4}. 

\begin{proof}[Proof of Theorem \ref{th4}]
By rearranging the  decomposing (\ref{eqq32}) and (\ref{eqq33}), it holds that
\begin{align*}
\Upsilon_{n+1}=(\I_H-\gamma_n \Gamma_{\bbeta})\Upsilon_n+\gamma_n\eta_{n+1}-\gamma_n r_{n1}-\gamma_n r_{n2}\,,
\end{align*}
where $\eta_n=\nabla G_{\bar{\bbeta}_n}( \Upsilon_n)+\bX_{n+1}(\Lambda_{n+1}- \bX_{n+1}^{T}\Upsilon_n )\|\Lambda_{n+1}-\bX_{n+1}^{T}\Upsilon_n \|^{-1}$, $r_{n1}=\Phi_{\bar{\bbeta}_n}(\Upsilon_n)- {\hat \Gamma}_{\bar{\bbeta}_n}{\Upsilon}_n $ and $r_{n2}=
({\hat \Gamma}_{\bar{\bbeta}_n}- \Gamma_{\bbeta}){\Upsilon}_n$.
To bound the remainder term $r_{n1}$ and $r_{n2}$, using the results from  Lemma \ref{lem3} and Lemma \ref{lem4}, we have $\|r_{n1}\|\leq C\|\Upsilon_n\|^2$
and $\|r_{n2}\|\leq C\|\bar{\bbeta}_n-\bbeta\|\|\Upsilon_n\|$. By some simple algebra with $r_k=r_{k1}+r_{k2}$,   we can obtain that for any $k\geq 1$
\[
\Gamma_{\bbeta} \Upsilon_k=\eta_{k+1}-r_k+\frac{1}{\gamma_k}(\Upsilon_k-\Upsilon_{k+1})\,.
\]
Summing  up these equalities,  we have
$
n\Gamma_{\bbeta}\bar{\Upsilon}_n=\sum_{k=1}^n\frac{1}{\gamma_k}(\Upsilon_k-\Upsilon_{k+1})-\sum_{k=1}^nr_k+\sum_{k=1}^n\eta_{k+1}\,.
$
Then it holds that with   $S_{n+1}=\sum_{k=1}^n\eta_{k+1}$,
\[
\sqrt{n}\Gamma_{\bbeta}\Upsilon_n=\frac{1}{\sqrt{n}}\left\{\frac{\Upsilon_1}{\gamma_1}-\frac{\Upsilon_{n+1}}{\gamma_n}+
\sum_{k=2}^n\Upsilon_k\left(\frac{1}{\gamma_k}-\frac{1}{\gamma_{k-1}}\right)\right\}-\frac{1}{\sqrt{n}}\sum_{k=1}^nr_k+
\frac{1}{\sqrt{n}}S_{n+1}\,. 
\]
To use the martingale CLT for  $\frac{1}{\sqrt{n}}S_{n+1}$,  we denote 
\[
\eta_n^{*}=\bX_{n+1}\frac{{\Lambda}_{n+1}}{\|{\Lambda}_{n+1}\|}=W_{n+1}\frac{\bX_{n+1}(Y_{n+1}-\bX_{n+1}\bar{\bbeta}_n)}{\|Y_{n+1}-\bX_{n+1}\bar{\bbeta}_n\|}\,,
\]
and write $\frac{1}{\sqrt{n}}S^{*}_{n+1}=\frac{1}{\sqrt{n}}\sum_{k=1}^n\eta_{k+1}^{*}$. According to  Corollary 6 in \citet{Lavrentyev2016}, conditional on observations, we have $\frac{1}{\sqrt{n}}S^{*}_{n+1}\xrightarrow{{\mathcal L}} N(0, \Omega_n)$, 
where $\Omega_n$ is defined as
\begin{align*}
\Omega_n=\frac{1}{n}\sum_{k=1}^n\frac{\bX_{k+1}\bX_{k+1}^{T}\otimes \{(Y_{k+1}-\bX_{k+1}{\bar \bbeta}_k)\otimes (Y_{k+1}-\bX_{k+1}{\bar \bbeta}_k)\}}{\|Y_{k+1}-\bX_{k+1}{\bar \bbeta}_k\|^2} \,.
\end{align*}
Let $\Omega_n^{*}$ denote $\frac{1}{n}\sum_{k=1}^n\frac{\bX_{k+1}\bX_{k+1}^{T}\otimes \{(Y_{k+1}-\bX_{k+1}{\bbeta})\otimes (Y_{k+1}-\bX_{k+1}{\bbeta})\}}{\|Y_{k+1}-\bX_{k+1}\bbeta\|^2}$, then it holds that
\begin{align*}
\|\Omega_n-\Omega_n^{*}\|\leq &\frac{1}{n}\sum_{k=1}^n\frac{C}{\|Y_{k+1}-\bX_{k+1}\bbeta\|}\|\bar{\bbeta}_k-\bbeta\|\\
\leq & C\left\{\frac{1}{n}\sum_{k=1}^n\frac{1}{\|Y_{k+1}-\bX_{k+1}\bbeta\|^2}\right\}^{1/2} \left\{\frac{1}{n}\sum_{k=1}^n\|\bar{\bbeta}_k-\bbeta\|^2\right\}^{1/2}\,,
\end{align*}
where the second inequality is implied by   the Cauchy--Schwarz inequality.  Due to $\mE\{\|U-h\|^{-2}\}\leq C$ imposed in Assumption \ref{as3} and $\|\bar{\bbeta}_n-\bbeta\|\rightarrow 0$ almost surely, it thus holds that $\|\Omega_n-\Omega_n^{*}\|\xrightarrow{P} 0$.
For any $u \in {\mathcal H}^d$,  we have
$\langle u, \Omega_{n}^{*}u \rangle\xrightarrow{P} \langle u, \bD\otimes B_0 u\rangle $ by law of large numbers.  Conditional on observations, it holds that 
\[
\sup_{t \in \mR}\bigg|\P(\langle u, G \rangle \leq t)- \P^{*}( \langle u, \frac{1}{\sqrt{n}}S_n^{*} \rangle \leq t)\bigg| \xrightarrow{p} 0
\] 
where $G \xrightarrow{{\mathcal L}} N(0, \bD\otimes B_{0})$.
Let $r_{k}^{*}=\eta_{k}-\eta_{k}^{*}$, and note that  
$\{r_{n}^{*}\}_{n\geq 1}$ is also a sequence of martingale differences, we have
\begin{align*}
\left\|\mE\left(r_k^{*}\otimes r_{k}^{*}| \F_n\right)\right\|\leq C\mE\{\|Y_{n+1}-\bX_n\bar{\bbeta}_n\|^{-1}| \F_n\}\|\Upsilon_n\|\,,
\end{align*}
Since $\|\Upsilon_n\|\rightarrow 0$ almost surely, 
it follows that $\left\|\mE\left(r_k^{*}\otimes r_{k}^{*}| \F_n\right)\right\| \rightarrow 0$ almost
surely. We then have $\frac{1}{\sqrt{n}}\sum_{k=1}^nr_{k+1}^{*}\xrightarrow{P} 0$.

In the following steps, we deal with the remainder terms.
As $\Upsilon_1$ is bounded almost surely,  it holds that $\frac{1}{\sqrt{n}} \| \frac{\Upsilon_1}{\gamma_1}\|\xrightarrow{p} 0$, and
$ \mE\{\|n^{-1/2} \Upsilon_{n+1} \gamma_n^{-1}\|^2\}\leq C n^{-1}n^{\alpha} \rightarrow 0$, with   Chebyshev's inequality 
$n^{-1/2} \Upsilon_{n+1} \gamma_n^{-1} \xrightarrow{p} 0$. Moreover, $|\gamma_k^{-1}-\gamma_{k-1}^{-1}|\leq 2\alpha \gamma k^{\alpha-1}$.  By   Lemma 5.5 of \cite{Cardot2017},  we have 
\begin{align*}
\frac{1}{n}\mE\left\|\sum_{k=2}^n\Upsilon_k\left(\frac{1}{\gamma_k}-\frac{1}{\gamma_{k-1}}\right)\right\|^2 \leq & \frac{1}{n}\bigg(\sum_{k=2}^n |\gamma_k^{-1}-\gamma_{k-1}^{-1}| \sqrt{\mE\|\Upsilon_k\|^2}\bigg)^2 \\
&~~~~~~~~~\lesssim \frac{1}{n}4\alpha^2\bigg(\sum_{k=2}^n k^{\alpha/2-1}\bigg)^2=\mO\left(\frac{1}{n^{1-\alpha}}\right)\,.
\end{align*}
We also have $n^{-1/2}\sum_{k=2}^n\Upsilon_k(\gamma_k^{-1}-\gamma_{k-1}^{-1}) \xrightarrow{p} 0$. Finally, it holds that
\begin{align*}
\frac{1}{n}\mE\left\|\sum_{k=1}^nr _k\right\|^2\leq & \frac{1}{n}\bigg(\sum_{k=1}^n\sqrt{\mE\|r_k\|^2}\bigg)^2\lesssim 
\frac{1}{n}\bigg(\sum_{k=1}^n\sqrt{\mE\|r_{k1}\|^2+\mE\|r_{k2}\|^2}\bigg)^2\\
&~~~~\leq \frac{1}{n}\bigg(\sum_{k=1}^n\sqrt{\mE\|\Upsilon_k\|^4+\mE\|\bar{\bbeta}_k-\bbeta\|^2\|\Upsilon_k\|^2}\bigg)^2
\leq \frac{1}{n}\bigg(\sum_{k=1}^n \frac{1}{k^{\theta/2}}\bigg)^2\\
&~~~~~~~=\mO\bigg(\frac{1}{n^{\theta-1}}\bigg)\,.
\end{align*}
Note that the third inequality relies on  Lemma \ref{lem3} and Lemma \ref{lem4},  and the last inequality is indicated by Lemma \ref{lem5}. As a result, we have $\frac{1}{\sqrt{n}}\sum_{k=1}^nr_k\xrightarrow{p} 0$, which completes the proof.

\end{proof}

\subsection{Proof of Theorem \ref{th5}}
\begin{proof}[Proof of Theorem \ref{th5}]
Let $\delta_{l}=\bbeta_n^r(t_l)-\bbeta(t_l)=\bar{\bbeta}_n(t_l)-\bbeta(t_l),  l=1,\ldots,m$, and  let $h$ be the linear interpolation of $\{t_{l}, \eta_l\}$, defined as 
\[
h(t)=\begin{cases}
\delta_1, \quad &0\leq t\leq t_1,\\
\frac{t_{l+1}-t}{t_{l+1}-t_l}\delta_{l}+\frac{t-t_l}{t_{l+1}-t_l}\delta_{l+1}, \quad &t_l\leq t\leq t_{l+1},\\
\delta_m,\quad &t_l\leq t\leq 1.
\end{cases}
\]
 Write $\bbeta_n^r(t)=Q_r(\bbeta(t)+h(t))$, where $Q_r$ is the operator associated with the $r$th order spline
interpolation. For a general function $f$,  it holds that $Q_r(f)$ is the solution to
\[
Q_r(f)=\argmin_{g\in W_2^r} \int \{g^{(r)}(t)\}^2 dt ~~~{\rm subject~ to~} g(t_{l})=f(t_{l}) (l=1,\ldots, m)\,.
\]
Since $Q_r$ is a linear operator,  we have $Q_r(\bbeta(t)+h(t))=Q_r(\bbeta)+Q_r(h)$, and 
\[
\|\bbeta_n^r-\bbeta\|_{2}\leq \|Q_r(\bbeta)-\bbeta\|_{2}+\|Q_r(h)\|_{2}\,.
\]
The approximation error led by the $r$th spline interpolation for $\bbeta$ can be bounded by \cite{Devore1993}, that is,
\[
\|Q_r(\bbeta)-\bbeta\|_{2}^2\lesssim m^{-2r}\,.
\]
Based on the  proof of Theorem \ref{th3},  it holds that $\|\bar{\bbeta}_n-\bbeta\|^2=\mO_p(1/n)$. Moreover,   by choosing ${\mathcal H}=L^2[0,1]$, we have $\int_0^1(\bar{\bbeta}_n(t)-\bbeta(t))\d t =O_p(1/n)$. Therefore, it follows that
\begin{align*}
\|Q_r(h)\|^2\lesssim \|h\|_{2}^2 \lesssim & m^{-1}\sum_{l=1}^m\delta^2_{l}=\frac{1}{m}\sum_{l=1}^m(\bar{\bbeta}_n(t_l)-\bbeta(t_l))^2\leq
\int_0^1(\bar{\bbeta}_n(t)-\bbeta(t))^2\d t\{1+o_p(1)\}\\
 &= \|{\bar\bbeta}_n-\bbeta\|_{2}^2\{1+o_p(1)\}=\mO_p\bigg(\frac{1}{n}\bigg)\,.
\end{align*}
\end{proof}





\end{appendices}



\begin{thebibliography}{99}
	

\bibitem[Bai et al.(1990)]{Bai1990} Bai, Z. D., Chen, X. R., Miao, B. Q. and Radhakrishna Rao, C. (1990). Asymptotic theory of least distances estimate in multivariate linear models. \textit{Statistics}, {\bf 21(4)}, 503--519.

\bibitem[Barber et al. (2017)]{Barber2017} Barber, R. F., Reimherr, M. and Schill, T. (2017). The function-on-scalar LASSO with applications to longitudinal GWAS. {\it Electronic Journal of Statistics}, {\bf 11},
1351--1389,

\bibitem[Bauer et al.(2018)]{Bauer2018}Bauer, A., Scheipl, F., K\"{u}chenhoff, H. and Gabriel, A. A. (2018). An introduction to semiparametric function-on-scalar regression. {\it Statistical Modelling}, {\bf 18(3-4)}, 346--364.

\bibitem[Cai and Hall (2006)]{Cai2006} Cai, T. and P. Hall (2006). Prediction in functional linear regression. 
\textit{The Annals of Statistics}, {\bf 34 (5)}, 2159--2179.

\bibitem[Cai and Yuan(2011)]{Cai2011} Cai,  T. and Yuan, M. (2011). Optimal estimation of the mean function based on discretely sampled functional data: phase transition. {\it The Annals of Statistics}, 39(5), 2330--2355.

\bibitem[Canay et al.(2021)]{Canay2021}Canay, I. A., Santos, A. and Shaikh, A. M. (2021). The wild bootstrap with a “small” number of “large” clusters. {\it Review of Economics and Statistics}, {\bf 103(2)}, 346--363.

\bibitem[Cardot et al.(2013)]{Cardot2013}  Cardot, H. and C\'{e}nac, P. and Zitt, P.A. (2013).Efficient and fast estimation of the geometric median in Hilbert spaces with an averaged stochastic gradient algorithm. 
 \textit{Bernoulli}, {\bf 19}, 18--43
 
\bibitem[Cardot and  Godichon-Baggioni(2017)]{CB2017} Cardot, H. and Godichon-Baggioni, A. (2017). Fast estimation of the median covariation matrix with application to online robust principal components analysis. \textit{Test}, {\bf 26(3)}: 461--480.
 
 \bibitem[Cardot et al.(2017)]{Cardot2017} Cardot, H., C\'{e}nac, P. and Godichon-Baggioni, A. (2017). Online estimation of the geometric median in Hilbert spaces: Nonasymptotic confidence balls. \textit{The Annals of Statistics}  {\bf 45}, 591--614.
 
 \bibitem[Chen et al.(2020)]{Chen2020} Chen, X., Lee, J. D. Tong, X. T., and Zhang, Y. (2020). Statistical inference for model
parameters in stochastic gradient descent. \textit{The Annals of Statistics}, {\bf 48(1)}:251--273.

\bibitem[Cheng et al.(2023)]{Cheng2023} Cheng, G., Peng, L. and Zou, C. (2023). Statistical inference for ultrahigh dimensional location parameter based on spatial median. arXiv preprint arXiv:2301.03126.


\bibitem[Devore and Lorentz(1993)]{Devore1993}Devore, R. A. and Lorentz, G. G.(1993) {\it Constructive approximation}, volume 303. Springer Science and Business Media.

\bibitem[Feng et al.(2018)]{Fang2018} Fang, Y., Xu, J. and Yang, L. (2018). Online bootstrap confidence intervals for the stochastic
gradient descent estimator. \textit{J. Mach. Learn. Res.}, {\bf 19(1)}:3053--3073.

\bibitem[Ghosal and Maity(2023)]{Ghosal2023}Ghosal, R. and Maity, A. (2023). Variable selection in nonlinear function‐on‐scalar regression. {\it Biometrics}, {\bf 79(1)}, 292--303.

\bibitem[Goldsmith et al.(2015)]{Goldsmith2015}Goldsmith, J., Zipunnikov, V. and  Schrack, J. (2015). Generalized multilevel function‐on‐scalar regression and principal component analysis. {\it Biometrics}, {\bf 71(2)}, 344--353.

\bibitem[Goldsmith and Kitago (2016)]{Goldsmith2016}Goldsmith, J. and T. Kitago (2016). Assessing systematic effects of stroke on motor control by using hierarchical function-on-scalar regression. {\it Journal of the Royal Statistical Society: Series
C (Applied Statistics)}, {\bf 65 (2)}, 215--236.


\bibitem[Godichon-Baggioni(2016)]{Godichon2016} Godichon-Baggioni, A. (2016). Estimating the geometric median in Hilbert spaces with stochastic gradient algorithms; $L^p$ and almost sure rates of convergence. \textit{J. Multivariate Anal.}, {\bf 146}, 209--222.

\bibitem[Godichon-Baggioni(2019)]{God2019} Godichon-Baggioni, A. (2019). Online estimation of the asymptotic variance for averaged stochastic gradient algorithms. \textit{J. Stat. Plan. Infer.}, {\bf 203}, 1--19.



\bibitem[Fan and Reimherr(2017)]{Fan2017}Fan, Z. and Reimherr, M. (2017). High-dimensional adaptive function-on-scalar regression. {\it Econometrics and statistics}, {\bf} 1, 167--183.

\bibitem[Hall et al.(2006)]{Hall2006} Hall, P. , Müller, H. and Wang, J. (2006). Properties of principal component methods for functional and longitudinal data analysis. {\it The Annals of Statistics}, 34(3).1493--1517. 


\bibitem[Jakubowski(1988)]{Jakubowski1988} Jakubowski, A. (1988). Tightness criteria for random measures with application to the principle of conditioning in Hilbert spaces. \textit{Probab. Math. Statist.}, {\bf 9(1)}:95--114.

\bibitem[Jhun and Choi(2009)]{Jhun2009} Jhun, M. and Choi, I. (2009). Bootstrapping least distance estimator in the multivariate regression model. \textit {Computational Statistics and Data Analysis}, {\bf 53(12)}, 4221--4227.

 

\bibitem[Kraus and Panaretos(2012)]{Kraus2012} Kraus, D. and Panaretos, V. M. (2012). Dispersion operators and resistant second-order functional data analysis. {\it Biometrika}, {\bf 99}, 813--832.

\bibitem[Lavrentyev and Nazarov(2016)]{Lavrentyev2016}Lavrentyev, V. and Nazarov, L.  (2016). A functional central limit theorem for hilbert-valued martingales. {\it Lobachevskii Journal of Mathematics}, {\bf 37},138--145.

\bibitem[Lee et al.(2022)]{Lee2022}
Lee, S., Liao, Y., Seo, M. H., and Shin, Y. (2022). Fast and robust online inference with stochastic gradient descent via random scaling.  {\it Proceedings of the AAAI Conference on Artificial Intelligence} {\bf 36}, 7381--7389).


\bibitem[Li et al.(2022)]{Li2022} Li, X., Liang, J., Chang, X., and Zhang, Z. (2022). Statistical estimation and online inference via local sgd.  {\it Conference on Learning Theory} {\bf 1}, 1613--1661).

\bibitem[Liu et al.(2022)]{Liu2020} Liu, Y., Li, M., and Morris, J. S. (2022). On Function-on-Scalar Quantile Regression. arXiv preprint arXiv:2002.03355.

\bibitem[Liu et al.(2022)]{Liu2022} Liu, R., Yuan, M., and Shang, Z. (2022). Online statistical inference for parameters estimation with linear-equality constraints.  \textit{J. Multivariate Anal.} {\bf 191}, 105017.

\bibitem[Minsker(2015)]{Minsker2015}Minsker, S.(2015). Geometric median and robust estimation in banach spaces. {\it Bernoulli,} {\bf 21(4)}, 2308.

\bibitem[Morris (2015)]{Morris2015} Morris, J. S. (2015). Functional regression. {\it Annual Review of Statistics and Its Application}, {\bf 2}, 321--359.

\bibitem[M\"{o}tt\"{o}nen, Nordhausen, and Oja(2010)]{Mottonen2010} M\"{o}tt\"{o}nen, J., Nordhausen, K., and Oja, H. (2010). Asymptotic theory of the spatial median. {\it Nonparametrics and Robustness in Modern Statistical Inference and Time Series Analysis} {\bf 7}, 182--193.

\bibitem[Nordhausen and  Oja(2011)]{Oja2011}Nordhausen, K. and Oja, H. (2011). Multivariate $L_1$ methods: The package MNM. \textit{Journal of Statistical Software} {\bf 43}, 1--28.

\bibitem[Oja(2010)]{Oja2010}
Oja, H. (2010) \textit{Multivariate nonparametric methods with R: An approach based on spatial signs and ranks}. Lecture Notes in Statistics, Springer, New York.

\bibitem[Padilla et al.(2022)]{Padilla2022}Padilla, O. H. M., W. Tansey, and Y. Chen (2022).  Quantile regression with ReLU Networks: Estimators and minimax rates, \textit{J. Mach. Learn. Res.}, {\bf 23(1)}:11251--11292.

\bibitem[Polyak and Juditsky(1992)]{Juditsky1992}Polyak, B. and Juditsky, A. (1992). Acceleration of stochastic approximation. \textit{SIAM J. Control and Optimization}, {\bf 30}:838--855.

\bibitem[Reiss et al.(2010)]{Reiss2010} Reiss, P. T.,  Huang, L. and Mennes, M. (2010). Fast function-on-scalar regression with penalized basis expansions. \textit{The International Journal of Biostatistics}, {\bf 6 (1)}.

\bibitem[Roberts, Mueller and Mclntyre(2017)]{Roberts2017} Roberts, D., Mueller, N. and McIntyre, A. (2017). High-dimensional pixel composites from earth observation time series. {\it IEEE Transactions on Geoscience and Remote Sensing}, {\bf 55}(11), 6254--6264.

\bibitem[Ramsay and Silverman (2005)]{Ramsay2005} Ramsay, J. and Silverman. B (2005). Functional Data Analysis. Springer-Verlag New York.

\bibitem[Rice and Silverman(1991)]{Rice1991} Rice, J. A. and Silverman, B. W. (1991). Estimating the mean and covariance structure nonparametrically when the data are curves. {\it Journal of the Royal Statistical Society. Series B: Methodological}, 53(1), 233--243.


\bibitem[Vardi and Zhang(2000)]{Vardi2000}Vardi, Y. and Zhang, C.-H. (2000). The multivariate $L_1$-median and associated data depth. 
\textit{Proc. Natl. Acad. Sci}, {\bf 97(4)},1423--1426.

\bibitem[Wang et al.(2017)]{Wang2017} Wang, X., H. Zhu, and A. D. N. Initiative (2017). Generalized scalar--on--image regression models via total variation. {\it Journal of the American Statistical Association}, {\bf 112 (519)}, 1156--1168.

\bibitem[Xie et al.(2023)]{Xie2023} Xie, J., Shi, E., Sang, P., Shang, Z., Jiang, B., and Kong, L. (2023). Scalable inference in functional linear regression with streaming data. arXiv preprint arXiv:2302.02457

\bibitem[Yang et al.(2019)]{Yang2019}Yang, H., Baladandayuthapani, V., Rao, A. U., and  Morris, J. S. (2019). Quantile function on scalar regression analysis for distributional data. {\it Journal of the American Statistical Association}.

\bibitem[Zhu and Dong(2021)]{Zhu2021}
Zhu, Y. and Dong, J. (2021). On constructing confidence region for model parameters in stochastic gradient descent via batch means. {\it 2021 Winter Simulation Conference} {\bf 1}, 1--12.

\bibitem[Zhang et al.(2022)]{Zhang2022}Zhang, Z., Wang, X., Kong, L. and Zhu, H. (2022). High-dimensional spatial quantile function-on-scalar regression. {\it Journal of the American Statistical Association},{\bf 117}, 1563--1578

\bibitem[Zhang et al.(2017)]{Zhang2017}Zhang, S., Guo, B., Dong, A., He, J., Xu, Z. and Chen, S. X. Cautionary tales on air-quality improvement in Beijing. {\it Proceedings of the Royal Society A: Mathematical,
Physical and Engineering Sciences} 473 (2205), 20170457.

\end{thebibliography}
\bibliographystyle{plainnat}

\end{document}